\documentclass[11pt, a4paper]{article}

\usepackage[a4paper,margin=2.5cm]{geometry}
\usepackage[utf8]{inputenc}
\usepackage{booktabs}
\usepackage[small]{caption}
\usepackage{lmodern}

\usepackage{graphicx}
\usepackage{tikz}
\usepackage{algorithm}
\usepackage{algpseudocode}
\usepackage{xspace}
\usepackage{mathtools}
\usepackage[super]{nth}
\usepackage{amsmath}

\usepackage{amsthm}
\usepackage{amsfonts}
\usepackage{thm-restate}
\usepackage{microtype}
\usepackage{hyperref}
\usepackage{enumitem}
\usepackage{multirow}
\DeclareCaptionFont{xbf}{\bfseries\boldmath}
\usepackage{tablefootnote}

\usepackage{svg}

\usepackage[font=scriptsize]{caption}
\usepackage{comment}

\usepackage{makecell}
\usepackage{pifont}
\newcommand{\cmark}{\ding{51}}%
\newcommand{\xmark}{\ding{55}}%
\usepackage{colortbl}

\usepackage[sort&compress,numbers]{natbib}
\usepackage{cleveref}
\algdef{SE}[SUBALG]{Indent}{EndIndent}{}{\algorithmicend\ }%
\algtext*{Indent}
\algtext*{EndIndent}

\algdef{SE}{Upon}{EndUpon}[1]{\textbf{upon} \(\mbox{#1}\) \textbf{do}}{\textbf{end upon}}%

\theoremstyle{plain}
\newtheorem{theorem}{Theorem}[section]
\newtheorem{definition}[theorem]{Definition}
\newtheorem{lemma}[theorem]{Lemma}
\newtheorem{corollary}[theorem]{Corollary}
\newtheorem{proposition}[theorem]{Proposition}

\theoremstyle{remark}
\newtheorem*{remark}{Remark}
\newtheorem{openquestion}[theorem]{Open question}

\newtheorem{observation}[theorem]{Observation}
\newtheorem{notation}[theorem]{Notation}

\newenvironment{sloppypar*}
 {\sloppy\ignorespaces}
 {\par}

\newcommand{\qedClaim}{\hfill \ensuremath{\Box}}

\newcommand{\byz}{\textrm{Byzantine parties}\xspace}
\newcommand{\byzadj}{\textrm{Byzantine}\xspace}

\newcommand{\safe}{\mathit{SafeArea}\xspace}
\newcommand{\mda}{\mathit{MDA}\xspace}
\newcommand{\barSet}{\ensuremath{\mathrm{S_{Cent}}\xspace}}
\newcommand{\trueBar}{\ensuremath{\mathrm{Cent}^{\star}}\xspace}
\newcommand{\encBall}{\ensuremath{\mathrm{Ball(\barSet)}}\xspace}
\newcommand{\encBalli}{\ensuremath{\mathrm{Ball_i(\barSet)}}\xspace}
\newcommand\ballOf[1]{\ensuremath{\mathrm{Ball(#1)}}\xspace}
\newcommand{\radiusEncBall}{\ensuremath{\mathrm{Rad\bigl(\encBall\bigr)}}\xspace}
\newcommand{\longestEdge}{\ensuremath{\mathcal{LE}}\xspace}
\newcommand{\convexHull}{\ensuremath{\mathrm{Conv}}\xspace}
\newcommand{\distance}{\ensuremath{\mathrm{dist}}\xspace}
\newcommand{\tb}{\ensuremath{\textrm{TB}}\xspace}
\newcommand{\tbi}{\ensuremath{\textrm{TB}_i}\xspace}
\newcommand{\bb}{\ensuremath{\textrm{CB}}\xspace}
\newcommand{\bbi}{\ensuremath{\textrm{CB}_i}\xspace}
\newcommand{\pbb}{\ensuremath{\widetilde{\textrm{CB}}}\xspace}
\newcommand{\pbbi}{\ensuremath{\widetilde{\textrm{CB}}_i}\xspace}
\newcommand{\centroidFunction}{\ensuremath{\mathrm{Centroid}}\xspace}
\newcommand{\midpointFunction}{\ensuremath{\mathrm{Midpoint}}\xspace}
\newcommand{\ExistingSafeArea}{\ensuremath{\mathrm{SafeAlg}}\xspace}
\DeclareMathOperator*{\argmin}{arg\,min}

\newcommand\restr[2]{{
\left.\kern-\nulldelimiterspace 
#1 
\vphantom{\big|} 
\right|_{#2} 
}}

\newcommand{\mytitle}[1]{
    \begingroup
    \fontsize{16pt}{18pt}\fontseries{bx}\selectfont \centering #1 \par  
    \endgroup
}
\newcommand{\myauthors}[1]{
    \begingroup
    \fontsize{12pt}{14pt}\fontseries{bx}\selectfont \centering #1 \par  
    \endgroup
}

\begin{document}

\mytitle{Centroid Approximation with Multidimensional Approximate Agreement Protocols}

\bigskip
\bigskip

\myauthors{M\'elanie Cambus\footnote{\texttt{melanie.cambus@aalto.fi}, Aalto University, Finland} and Darya Melnyk\footnote{\texttt{melnyk@tu-berlin.de}, TU Berlin, Germany}}

\bigskip
\bigskip

\begin{abstract}

In this paper, we present distributed fault-tolerant algorithms that approximate the centroid (i.e., the average) of a set of $n$ data points in $\mathbb{R}^d$. Our work falls into the broader area of multidimensional Byzantine approximate agreement. 
We show that state-of-the-art algorithms, such as agreeing inside the convex hull of all non-faulty vectors, or minimum-diameter averaging (MDA), in the worst case either prevent us from agreeing on a vector close to the centroid (in terms of approximation quality), or allow Byzantine parties to influence the output considerably (in terms of validity).

To design better approximation algorithms, we propose a novel concept of defining an approximation ratio of the centroid by including the vectors of the \byzadj adversaries in the definition. We analyze the algorithms in the synchronous and asynchronous models of communication with public communication channels.
We show that the standard agreement algorithms based on agreeing inside the convex hull of all non-faulty vectors do not allow us to compute a better approximation than $2d$ of the centroid. 
On the other hand, MDA can be used to achieve constant approximation at the cost of only satisfying strong validity.
As a trade-off, we develop an approach that reaches a $2\sqrt{d}$-approximation of the centroid, while satisfying box validity. Our approach provides optimal resilience, allowing up to $t<n/3$ faulty nodes. 

\end{abstract}

\section{Introduction}

Multidimensional Byzantine approximate agreement (MBAA) is an important subroutine that allows nodes in a network to approximately agree on vectors that are close to a convergence vector, in the presence of any type of node failures, from crashing to organized malicious behavior.  
The convergence vector should typically satisfy a so-called validity condition. This condition makes sure that the convergence vector is not just a trivial predetermined vector, but depends on the inputs of the non-faulty nodes.
MBAA is particularly interesting for practical applications where the fast running time of an agreement algorithm is more important than agreement on the same vector. 
In practical applications, there is often a desired convergence vector that the system should agree on, for example, the centroid, the weighted average, or the geometric median.
This work considers the centroid as the desired convergence vector, as it is a well-known and broadly used representative vector in many applications such as vector quantization~\cite{vectorQuantization}, collaborative learning~\cite{NEURIPS2021_d2cd33e9}, data anonymization~\cite{fiorina_plott_1978}, large-scale elections~\cite{7474137}, community detection algorithms~\cite{li2022centroid,lloyd1982least} and distributed voting~\cite{mendes2015multidimensional}.

Different approaches have been considered to evaluate the quality of the output of an MBAA algorithm:
The traditional approach is to classify the output using validity conditions.   
However, even the most restrictive validity condition introduced for MBAA---the convex validity condition~\cite{mendes2015multidimensional}---only guarantees that the convergence vector is in the convex hull of all input vectors of the non-faulty nodes. 
There are two reasons why this polytope does not represent the centroid well. If the centroid is inside the polytope, the polytope can be large such that the convergence vector is far from the centroid of the non-faulty input vectors. If the polytope is small, the polytope itself may lie far away from the centroid of the non-faulty input vectors.
Recently, the absolute distance to the centroid compared to the maximum distance between any two non-faulty input vectors has been proposed to evaluate the quality of MBAA algorithms~\cite{NEURIPS2021_d2cd33e9}. This measure only captures the worst-case scenario where Byzantine inputs are the only outliers in the data. If one single non-faulty outlier is present in the data, the measure cannot guarantee that the output is closer to the centroid than the distance to this outlier. The same holds when no Byzantine parties are present in the system.  

In this paper, we provide a novel definition for centroid approximation (see~\Cref{subsec: approx-of-centroid}). The approximation is defined relative to the input distribution: if the input distribution is unfavorable (e.g., there is one outlier among the non-faulty nodes), even for an optimal algorithm, it is impossible to differentiate between a non-faulty outlier and a Byzantine node. On the other hand, if the input distribution is favorable (e.g., the Byzantine party shares a similar vector to the non-faulty inputs), an optimal algorithm should be capable of agreeing on a vector close to the centroid. Our approximation definition compares an approximate agreement algorithm to what an optimal algorithm that cannot detect Byzantine behavior can achieve for a given input distribution. 
To make sure that an adversary cannot change the convergence vector arbitrarily, we make use of validity conditions. We thus focus on finding a trade-off between the quality of the approximation of the non-faulty nodes and the validity condition. 


We focus on the following validity conditions in this work (fully defined in \Cref{subsec:multidim-approx-agreement} and \Cref{subsec: boxes-definitions}), presented from loose to more restrictive conditions:
The \textbf{weak validity} condition~\cite{civit2022byzantine,civit2021polygraph,yin2019hotstuff}, where if there are no failures and all input vectors are identical, then the output vector of each node is required to be the input vector; The \textbf{strong validity} condition~\cite{63511,abraham_et_al:LIPIcs.DISC.2017.41,civit2022byzantine, 8548057}, where if all non-faulty nodes have the same input vector then this input vector must be the output vector of all non-faulty nodes; The \textbf{box validity} condition~\cite{intervalValidity}, where the output vector of each non-faulty node must be in the smallest coordinate-parallel
box containing all input vectors of non-faulty nodes; The \textbf{convex validity} condition~\cite{fgger_et_al:LIPIcs:2018:9816,7867756,abbas2022centerpoint,wang2019computingTverbergPoint,attiya_et_al:LIPIcs.OPODIS.2022.6,mendes2015multidimensional}, where the output vector of each non-faulty node must be in the convex hull of all input vectors of non-faulty nodes.

Our contributions are as follows: 
\begin{itemize}
    \item We prove that any algorithm guaranteeing the convex validity condition has an approximation ratio of at least $2d$ (\Cref{sec:bad_approx_safe_area}), 
    \item we prove that there exists an algorithm guaranteeing the weak validity condition that has a tight approximation of the centroid (\Cref{sec:synch_one_approximation}),
    \item we prove that there exists an algorithm guaranteeing the strong validity condition that has a tight approximation of the centroid\footnote{we show that the approximation is tight in the synchronous setting, but there remains a $0.8$ gap in the asynchronous setting} (\Cref{sec:synch_const_approximation}),
    \item and finally, as a trade-off  between the above solutions, we give an algorithm guaranteeing box validity condition that has an approximation ratio of at most $2\sqrt{d}$, giving at least a quadratic improvement compared to the $\safe$ algorithms from~\cite{mendes2015multidimensional} in terms of approximation of the centroid while sacrificing very little on the validity condition (\Cref{sec:synch_sqrt_d_approximation}).
\end{itemize}

Note that the resilience (i.e., \ maximum amount of tolerated \byz) of our algorithm is optimal in both synchronous and asynchronous settings, contrary to the other studied algorithms. Finally, our algorithm has only polynomial local computation steps, whereas the algorithms described and analyzed in \Cref{sec:bad_approx_safe_area}, \Cref{sec:synch_one_approximation} and \Cref{sec:synch_const_approximation} require exponential local computations. Our main result is shown in \Cref{thm:syncBoxAlg}. 
We present a summary of the results of this paper in Table~\ref{tab:overview}.

\begin{remark}
Observe that the proven lower bounds for different validity conditions would also hold for exact \byzadj agreement, which is beyond the scope of this paper.    
\end{remark}

\newcommand{\CC}[1]{\cellcolor{gray!#1}}
\begin{table}[h!]
\centering
\begin{tabular}{ c c c c c c c c c } 
  & & \multicolumn{4}{ c }{Validity} & Resilience & \multicolumn{2}{ c }{Approximation} \\   
  \cline{3-6}
  \cline{8-9}
  &  & \CC{1}\rotatebox[origin=c]{90}{\parbox[c]{1.2cm}{\centering weak}} & \CC{5}\rotatebox[origin=c]{90}{\parbox[c]{1.2cm}{\centering strong}} & \CC{1}\rotatebox[origin=c]{90}{\parbox[c]{1.2cm}{\centering box}} & \CC{5}\rotatebox[origin=c]{90}{\parbox[c]{1.2cm}{\centering convex}} &  & \CC{5}sync. & \CC{1}async. \\ 
 \cline{2-9}
 \multicolumn{1}{l}{\multirow{12}{*}{\rotatebox[origin=c]{90}{Method}}} 
 & \CC{15}\Gape[0pt][2pt]{\makecell{ $\safe$ Algorithms\\ \cite{fgger_et_al:LIPIcs:2018:9816,mendes2015multidimensional}\ \scriptsize{(\Cref{thm:unbounded_sync_safe_area})}}} & \CC{5}\cmark & \CC{15}\cmark & \CC{5}\cmark & \CC{15}\cmark & \CC{5}$t<\min\left(\frac{n}{3}, \frac{n}{d+1}\right)$ & \CC{15}$\infty$ & \CC{5}$\infty$ \\ 
 & \CC{5}\Gape[0pt][2pt]{\makecell{$\safe$ Approaches\tablefootnote{We refer to agreement inside the $\safe$ (defined in \Cref{subsec:multidim-approx-agreement}) as the $\safe$ approaches. } \\ \scriptsize{(\Cref{thm:sync_safe_area})}}} & \CC{1}\cmark & \CC{5}\cmark & \CC{1}\cmark & \CC{5}\cmark & \CC{1}$t<\min\left(\frac{n}{3}, \frac{n}{d+1}\right)$ & \CC{5}$\Omega(d)$ & \CC{1}$\Omega(d)$ \\ 
 & \CC{15}\Gape[0pt][2pt]{\makecell{Centroid \& $\safe$ \\ \scriptsize{(\Cref{thm:synch-one-approx,thm:asynch_two_approx})}}} & \CC{5}\cmark & \CC{15}\xmark & \CC{5}\xmark & \CC{15}\xmark & \CC{5}$t<\min\left(\frac{n}{3}, \frac{n}{d+1}\right)$ & \CC{15}$1$ & \CC{5}$2$ \\
 & \CC{5}\Gape[0pt][2pt]{\makecell{$\mda$~\cite{NEURIPS2021_d2cd33e9} \\ 
 \scriptsize{(\Cref{thm:mda_correctness,thm:asynch-mda}})}} & \CC{1}\cmark & \CC{5}\cmark & \CC{1}\xmark & \CC{5}\xmark & \CC{1}$t<\frac{n}{4}$ and $t<\frac{n}{7}$& \CC{5}$3.8$ & \CC{1}10.4 \\ 
 & \CC{15}\Gape[0pt][2pt]{\makecell{$\mda$ \& $\safe$ \\ \scriptsize{(\Cref{sec:synch_const_approximation} and \ref{obs:asynch_mda_and_safe})}}} & \CC{5}\cmark & \CC{15}\cmark & \CC{5}\xmark & \CC{15}\xmark & \CC{5}$t<\min\left(\frac{n}{7}, \frac{n}{d+1}\right)$& \CC{15}$2$ & \CC{5}$4.8$\\ 
 & \CC{5}\Gape[0pt][2pt]{\makecell{Box Algorithm \\ \scriptsize{(\Cref{thm:syncBoxAlg,thm:asynch_BoxAlgo})}}} & \CC{1}\cmark & \CC{5}\cmark & \CC{1}\cmark & \CC{5}\xmark & \CC{1}$t<\frac{n}{3}$ & \CC{5}$O\bigl(\sqrt{d}\bigr)$ & \CC{1}$O\bigl(\sqrt{d}\bigr)$\\ 
 & \CC{15}\Gape[0pt][2pt]{\makecell{$RB\!-\!TM$~\cite{NEURIPS2021_d2cd33e9} \\ \scriptsize{(\Cref{cor:synch-RB-TM,cor:asynch-RB-TM})}}} & \CC{5}\cmark & \CC{15}\cmark & \CC{5}\cmark & \CC{15}\xmark &  \CC{5}$t<\frac{n}{3}$ & \CC{15}$O\bigl(\sqrt{d}\bigr)$ & \CC{5}$O\bigl(\sqrt{d}\bigr)$ \\
\end{tabular}
\caption{Overview of the results of this paper. 
Please note that methods based on $\mda$ and $\safe$ rely on exponential local computation time, while $RB\!-\!TM$ and the Box Algorithm can be computed in polynomial time. 
}
\label{tab:overview}
\end{table}

\paragraph*{Technical overview.} 
The first technical challenge is to fix a meaningful definition of approximation ratio in the presence of \byz.  
A traditional way of defining the approximation ratio of an algorithm would be to compute the ratio between the algorithm's result and the result of an optimal algorithm. Note that with such an absolute measure, a \byzadj party can change the centroid arbitrarily. 
 The best possible upper bound on the absolute distance is the maximum distance between any two input vectors (see~\cite{NEURIPS2021_d2cd33e9}). 
Instead, we use the fact that even an optimal algorithm cannot always identify \byz to our advantage: the true centroid can be any centroid computed from at least $n-t$ received vectors (as up to $t$ nodes can be \byzadj). 
An optimal algorithm that cannot identify \byz has to minimize the maximum distance to all of those possible centroids. This is a particular case of the Minimax problem, whose solution is to find the center of the smallest enclosing ball, i.e.\ the center of the smallest ball that contains all possible centroids. 
Our definition of approximation directly follows from this and is given in \Cref{subsec: approx-of-centroid}.

Another technical challenge is to get a significant improvement in the centroid approximation while satisfying a sound validity condition. Since we show in \Cref{sec:bad_approx_safe_area} that satisfying the convex validity condition makes the approximation at least $2d$, we aim for the next best validity condition: the box validity condition. 
We define the smallest box containing all non-faulty input vectors (in which the output vectors of non-faulty nodes need to be), but also the smallest box containing all possible centroids of $n-t$ vectors (\Cref{subsec: boxes-definitions} and \Cref{sec:synch_sqrt_d_approximation}). We call those boxes the trusted box and the centroid box respectively, and show that their intersection is non-empty. Agreeing in the intersection of those boxes guarantees the approximation ratio to be at most $2\sqrt{d}$, and also guarantees the box validity condition.
However, those boxes are not computable locally due to the presence of \byz: nodes do not know which vectors come from \byz. We consider that \byz can send different vectors to different non-faulty nodes at each round, and can also not send any vector (see \Cref{sec: model-and-def}). In the asynchronous setting, the additional challenge is that \byz can also block messages from non-faulty nodes, making it possible for each non-faulty node at each round to receive up to $t$ \byzadj vectors even though it only waits for $n-t$ vectors.
To solve these issues, we define for each node a local trusted box and a local centroid box, which are guaranteed to be contained in their global counterpart. We also show these local boxes always intersect, making it possible for our algorithm to converge (see \Cref{def:multidim_agreement_properties}). The details of this are given in the proof of \Cref{thm:syncBoxAlg}.

\section{Related Work}

Approximate \byzadj agreement has been introduced for the one-dimensional case as a way to agree on arbitrary real values inside a bounded range~\cite{ApproximateAgreement}. 
Approximate agreement algorithms can speed up standard \byzadj agreement algorithms and are thus interesting for practical applications. 
Observe that any algorithm for solving synchronous \byzadj agreement exactly requires at least $t+1$ communication rounds~\cite{FischerLynchMinRounds}. In the asynchronous case, it is impossible to solve \byzadj agreement deterministically already if one faulty (e.g., \byzadj) node is present~\cite{FLPimpossibility}. In contrast, a synchronous or asynchronous approximate agreement algorithm converges to a solution in $O(\log(\delta(V)/\varepsilon))$ rounds, such that the values of the non-faulty nodes are inside an interval of size $\varepsilon$, where $\delta(V)$ denotes the maximum distance between any two input values. 
The first approximate \byzadj agreement algorithms~\cite{ApproximateAgreement,ApproximateAgreement2} did not satisfy the optimal resilience of $t<n/3$. This bound was first achieved by~\cite{DolevFIFObroadcast} in the synchronous, as well as asynchronous setting by using reliable broadcast~\cite{BrachaRB, TouegRB}. 

The multidimensional version of approximate agreement was first introduced by Mendes and Herlihy~\cite{VectorConsensusAsynch} and Vaidya and Garg~\cite{VectorConsensus}. These results were later combined in~\cite{mendes2015multidimensional}. The focus of this line of work is to establish agreement on a vector that is relatively ``safe''. In particular, the vector should lie inside the convex hull of all non-faulty vectors. This can be achieved through the Tverberg partition~\cite{https://doi.org/10.1112/jlms/s1-41.1.123}. The authors provide synchronous and asynchronous algorithms to solve exact and approximate agreement in this setting. This approach can only tolerate up to $t<\min\{n/3,n/(d+1)\}$ \byzadj nodes.
This algorithm idea has recently received much attention in different distributed models~\cite{fgger_et_al:LIPIcs:2018:9816,7867756,abbas2022centerpoint,wang2019computingTverbergPoint,attiya_et_al:LIPIcs.OPODIS.2022.6}. 
Also, other variants of multidimensional approximate agreement have been proposed in the literature, such as $k$-relaxed vector consensus~\cite{xiang_et_al:LIPIcs:2017:7095} and agreement in the validated setting~\cite{dotan2022validated}. 

Our work deviates from the previous work on approximate multidimensional agreement in the sense that we relax the strong validity property of agreeing inside the convex hull of all non-faulty vectors. To this end, we use the box validity condition that to date has been regarded as weak or impractical~\cite{abbas2022centerpoint,mendes2015multidimensional}. A similar goal has been pursued by El-Mhamdi, Farhadkhani, Guerraoui, Guirguis, Hoang, and Rouault~\cite{NEURIPS2021_d2cd33e9}. The authors focus on collaborative learning under \byzadj adversaries and reduce the corresponding problem to averaging agreement. In averaging agreement, the task is to agree on a vector close to the centroid. However, the computed solution is bounded using the distance between the two farthest vectors from non-faulty nodes, and no focus is put on the validity of the solution. 

Our work studies validity conditions for multidimensional Byzantine agreement protocols, with the goal of finding a trade-off between the quality of the solution and the validity requirement. One of the first studies of the quality of \byzadj agreement protocols was presented by Stolz and Wattenhofer~\cite{MedianValidity}. The authors consider the one-dimensional multivalued \byzadj agreement and provide a two-approximation of the median of the non-faulty nodes. Melnyk and Wattenhofer~\cite{intervalValidity} improve the two-approximation result for the median computation to output an optimal approximation. The same paper also introduces the box validity property for multidimensional synchronous agreement used in this work. 
The first paper discussing validity conditions for multidimensional \byzadj agreement focused on the application of \byzadj agreement in voting protocols~\cite{ByzantinePreferentialVoting}. There, a validity condition was introduced that satisfies unanimity among non-faulty voters. This validity condition is also satisfied by the convex validity condition in multidimensional approximate agreement~\cite{mendes2015multidimensional}. Allouah, Guerraoui, Hoang, and Villemaud~\cite{allouah2022robust} extend the unanimity condition to sparse unanimity, but they only focus on centralized \byzadj-resilient voting protocols. Civit, Gilbert, Guerraoui, Komatovic, and Vidigueira~\cite{civit2023validity} initiate a first study on the impact of validity conditions on the solvability of Byzantine consensus in the partially synchronous communication model.

\section{Model and Definitions}\label{sec: model-and-def}

We consider a networked system of $n$ nodes, each holding an input vector $v\in \mathbb{R}^d$. 
Note that, in order for nodes in the network to exchange vectors, we assume that a coordinate system is fixed, and none of the studied or presented algorithms proceed to changing this coordinate system. 
We assume the standard Euclidean structure on $\mathbb{R}^d$, and use the prevailing $\Vert\cdot \Vert_2$ norm.
For two vectors $x=(x_1, \dots, x_d)$ and $y=(y_1,\dots,y_d)$, we measure the distance in terms of the Euclidean distance: 
$\distance(x, y) = \sqrt{\sum_{i=1}^d(x_i-y_i)^2} = \Vert x-y\Vert_2. $

The nodes of the network can communicate with each other in a fully connected peer-to-peer network via public channels. We assume that up to $t<n/3$ nodes can show \byzadj behavior, i.e.\ such nodes can arbitrarily deviate from the protocol and are assumed to be controlled by a single adversary. Non-faulty nodes we call correct.
Similar to~\cite{mendes2015multidimensional}, we require that the communication between the nodes is reliable. 
Reliable broadcast has been introduced by Bracha~\cite{BrachaRB} and Srikanth and Toueg~\cite{TouegRB}, and it is discussed in more detail in~\cite{attiya2004distributed, cachin2011introduction}. It satisfies the following two properties while tolerating up to $t<n/3$ \byzadj nodes:
\begin{itemize}
    \item If a correct node broadcasts a message reliably, all correct nodes will accept this message.
    \item If a \byzadj node broadcasts a message reliably, all correct nodes will either accept the same message or not accept the \byzadj message at all.
\end{itemize}
The second condition of reliable broadcast is only satisfied eventually. Therefore, when using the reliable broadcast as a subroutine in agreement protocols, we can only guarantee the following condition:
\begin{itemize}
    \item If two correct nodes accept a message from a \byzadj party, the accepted message must be the same. 
\end{itemize}
The latter condition is also referred to as the consistency condition in the literature~\cite{Cachin:2014:IRS:2755417}, and the corresponding broadcast routine as consistent broadcast. 
Observe that the reliable broadcast routine from \cite{BrachaRB} applied in our protocols satisfies all the above conditions.

We consider two types of communication: synchronous and asynchronous. In synchronous communication, the nodes communicate in discrete rounds. In each round, a node can send a message, receive messages, and perform some local computation. A message sent by a correct node in round $i$ arrives at its destination in the same round. In asynchronous communication, the delivery time of a message is unbounded, but it is guaranteed that a message sent by a correct node arrives at its destination eventually. Since messages can be arbitrarily delayed, it is not possible to differentiate between \byz who did not send a message, and correct nodes whose messages were delayed until other nodes terminated. Observe that it is possible to simulate rounds in this model by letting the nodes attach a local round number to their messages. To make sure that the nodes make progress, it is assumed that a local round consists of sending a message to all nodes in the network, receiving $n-t$ messages with the same round count, and performing some local computation.

\subsection{Multidimensional approximate agreement}\label{subsec:multidim-approx-agreement}

In this work, we design deterministic algorithms that solve multidimensional approximate agreement; every node thereby executes the same algorithm:
\begin{definition}[Multidimensional approximate agreement]\label{def:multidim_agreement_properties}
    Given $n$ nodes, up to $t$ of which can be \byzadj, the goal is to design a deterministic distributed algorithm that satisfies:
    \begin{description}[noitemsep,topsep=3pt]
        \item[$\varepsilon$-Agreement:] Every correct node decides on a vector s.t.\ any two vectors of correct nodes are at a Euclidean distance of at most $\varepsilon$ from each other.
        \item[Strong Validity:] If all correct nodes started with the same input vector, they should agree on this vector as their output.
        \item[Termination:] Every correct node terminates after a finite number of rounds.
    \end{description}
\end{definition}

The strong validity condition was originally introduced in the definition of the consensus problem~\cite{10.1145/800221.806706,10.1007/BFb0040405,BrachaRB} and it is widely used in the literature~\cite{63511,abraham_et_al:LIPIcs.DISC.2017.41,civit2022byzantine, 8548057}. Sometimes, a weaker version of validity is considered in the literature~\cite{civit2022byzantine,civit2021polygraph,yin2019hotstuff}, called the weak validity condition. This validity condition has been used for $k$-set agreement~\cite{10.1145/301308.301368} and BFT protocols~\cite{civit2021polygraph, yin2019hotstuff}. 

\begin{description}[noitemsep,topsep=0pt]
    \item[Weak Validity:] If all nodes are correct and start with the same input vector, they should agree on this vector as their output.
\end{description}

In many state-of-the-art multidimensional approximate agreement protocols~\cite{fgger_et_al:LIPIcs:2018:9816,7867756,abbas2022centerpoint,wang2019computingTverbergPoint,attiya_et_al:LIPIcs.OPODIS.2022.6,mendes2015multidimensional}, the validity condition is instead replaced by the more restrictive convex validity condition:

\begin{description}[noitemsep,topsep=0pt]
    \item[Convex Validity:] Each correct node decides on a vector that is inside the convex hull of all correct input vectors. 
\end{description}

To guarantee convex validity in \byzadj agreement protocols, Mendes, Herlihy, Vaidya, and Garg~\cite{mendes2015multidimensional} have the nodes terminate on a vector inside the so-called $\safe$:

\begin{definition}[$\safe$]\label{def: safeArea}
The $\safe$ of a set of vectors $\{v_i, i\in \left[n\right]\}$ that can contain up to $t$ \byzadj vectors is defined as, 
    $$\safe = \bigcap_{\substack{I\subseteq \left[n\right]\\ |I|=n-t}}\convexHull\bigl(\{v_i, i\in I\}\bigr),$$
where $\convexHull\bigl(\{v_i, i\in I\}\bigr)$ is the convex hull of the set $\{v_i, i\in I\}$.
\end{definition}

The process of agreeing inside the $\safe$ is costly. In \cite{mendes2015multidimensional}, the authors showed that the upper bound on the number of \byz must be $t<\min(n/3, n/(d+1))$ in the synchronous case and $t<n/(d+2)$ in the asynchronous case.  

The goal of this paper is to establish approximate agreement on a representative vector: the centroid of the input vectors of all correct nodes. The centroid is defined as follows:

\begin{definition}[Centroid]
    The centroid of a finite set of $n$ vectors $\{v_i, i\in \left[n\right]\}$ is $\frac{1}{n}\sum_{i=1}^n v_i$.
\end{definition}

We use $\trueBar$ to denote the centroid of the correct input vectors. We also refer to the set of correct input vectors as the \emph{true vectors} and to $\trueBar$ as the \emph{true centroid} respectively.

In Section~\ref{sec:bad_approx_safe_area}, we show that the restriction of the strong validity condition to the convex validity condition cannot provide a good approximation of $\trueBar$. We will therefore relax convex validity in order to improve both the approximation of the centroid and the resilience of the agreement protocols.

\subsection{Approximation of the centroid}\label{subsec: approx-of-centroid}

The correct nodes cannot agree on the true centroid if \byzadj parties are present in the system. Let $f\le t$ denote the actual number of \byzadj parties that are present. A \byzadj party can follow the protocol while choosing worst-case input vectors and thus be indistinguishable from a correct node. While it changes the outcome, such \byzadj behavior is not detectable, and even an optimal approach cannot find out whether all $t$ nodes or only a subset of them are \byzadj. We therefore define the approximation ratio of a true centroid based on all $n$ input vectors for the worst case when $f=t$.  

Given all $n$ input vectors, consider all subsets of $n-t$ vectors. At least one of these contains only true vectors. In the case $f<t$, there can be multiple such sets, while in the case $f=t$ there is only one. 
As mentioned earlier, we cannot determine which of these sets is a set of correct nodes. Thus, each subset of $n-t$ vectors could potentially be the (only) subset containing only true vectors. In the following, we define the set of all possible centroids of $n-t$ nodes, i.e., candidates for being $\trueBar$ assuming the worst case where $f=t$ nodes are \byzadj.

\begin{definition}[Set of possible centroids]
    The set containing all possible centroids of $n-t$ nodes is denoted $\barSet$  and defined as $$\barSet\coloneqq \Bigl\{\frac{1}{n-t}\sum_{i\in I} v_i\ \big|\ \forall I\subset \left[n\right] \text{ s.t.\ } |I|=n-t \Bigr\}. $$ 
\end{definition}

Observe that, in practice, $\trueBar$ is not necessarily one of the elements of $\barSet$, as the centroids in $\barSet$ might have been computed from fewer than $n-f$ vectors. We can however make the following observation, which allows us to use $\barSet$ also in the case $f<t$:

\begin{observation}\label{obs: truebar in enc ball}
    $\trueBar\in\convexHull(\barSet)$.
\end{observation}
For $f = t$, the observation simply follows from the fact that $\trueBar\in \barSet$. In other cases, let $\{v_i\mid i\in [n-f]\}$ denote the set of true vectors for a better overview. The true centroid is defined as $\trueBar = \frac{1}{n-f}\sum_{i\in[n-f]} v_i$. Note that the sum of centroids from $\barSet$ that only contain true vectors is a multiple of $\trueBar$ (indeed, we consider all possible sets of $n-t$ vectors, implying that all true vectors appear in exactly the same number of sets). Therefore, $\trueBar$ can be written as a convex combination of elements of $\barSet$ and it is inside $\convexHull(\barSet)$.

In order to compute the best-possible approximation of the true centroid w.r.t.\ the Euclidean distance, we need to determine a vector that minimizes the maximum distance to all computed centroids, or as centrally as possible in $\convexHull(\barSet)$. Finding this point corresponds to solving a Minimax problem known as the $1$-center problem for the set $\barSet$, which requires finding the smallest enclosing ball of $\barSet$ and computing its center.

\begin{proposition}[Best-possible centroid approximation vector]
    Given $n$ input vectors, $f$ of which are \byzadj (where $f=t$ is the worst case). The best possible approximation vector of the true centroid $\trueBar$ in the worst case is the center of the smallest enclosing ball\, $\encBall$ of\ \,$\barSet$.
\end{proposition}

Given that the best-possible approximation of the true centroid is the center of $\encBall$, and $\trueBar$ could be on its boundary, the distance between the best-possible approximation and the true centroid can be up to $\radiusEncBall$. Using this idea, we can now define the optimal centroid approximation for any algorithm as follows:

\begin{definition}[Optimal centroid approximation]\label{def:one-approx}
    Let\, $\radiusEncBall$ be the radius of $\encBall$.
    All vectors at a distance of at most\, $\radiusEncBall$ from $\trueBar$ are defined to provide an optimal approximation of the true centroid. 
\end{definition}

We can now define the approximation ratio of an algorithm for computing the centroid depending on the radius of the smallest enclosing ball $\radiusEncBall$:

\begin{definition}[$c$-approximation of the centroid]\label{def:approx}
    Let $\mathcal{A}$ be an algorithm computing an approximation vector of the true centroid. Let $O_{\mathcal{A}}$ be the output of the algorithm on some input $I$. The approximation ratio of an algorithm $\mathcal{A}$ on input $I$ is defined as the smallest $c$ for which holds 
    $$\distance(O_{\mathcal{A}},\trueBar) \le c\cdot\radiusEncBall.$$
    We further say that $\mathcal{A}$ computes a $c$-approximation of the true centroid if, for any admissible input to the algorithm, the approximation ratio of the output of $\mathcal{A}$ is upper bounded by $c$.
\end{definition}

Note that admissible inputs consist in each of the $n$ nodes having an input vector in $\mathbb{R}^d$, and at most $t$ of these nodes being $\byzadj$. 

\subsection{Weakening the convex validity condition}\label{subsec: boxes-definitions}

In addition to centroid approximation, we care about validity conditions. We strive for a more restrictive one than the strong validity from \Cref{def:multidim_agreement_properties}. 
However, the $\safe$ approaches do not give a good approximation of the centroid, which makes the convex validity condition too restrictive (\Cref{sec:bad_approx_safe_area}). 
To achieve a better approximation ratio without sacrificing too much on the validity condition, we relax the convex validity condition to the ``coordinate-parallel box'' of true vectors:

\begin{definition}[Trusted box]
    Let $f\leq t$ be the number of \byzadj nodes and let $v_i$, $i \in [n-f]$ denote the true vectors. Let $v_i[k]$ denote the $k^{th}$ coordinates of these vectors. 
    The \emph{trusted box} $\tb$ is the Cartesian product of  
    $\bigl[\min_{i\in [n-f]}v_i[k], \max_{i\in [n-f]}v_i[k]\bigr], $ for $k\in[d]$.
    Denote $\tb[k]$ the orthogonal projection of $\tb$ onto the $k^{th}$ coordinate.
\end{definition}

In other words, $\tb$ is the smallest box containing all true vectors. 
Given this definition, we can define a relaxed version of the convex validity condition, referred to as the box validity condition in the literature~\cite{intervalValidity}:

\begin{description}[noitemsep,topsep=0pt]
    \item[Box validity:] The output vectors of the correct nodes at termination should be inside the trusted box $\tb$.
\end{description}
The trusted box cannot be locally computed by the correct nodes if the \byz submit their vectors.  
We hence define a locally trusted box that is computable locally by each correct node. The following definition is designed for synchronous communication, and is adapted for asynchronous communication in \Cref{sec:asynch_approximations}. 
We use $M_i$ to denote the set of messages received by node $i$ in a communication round. Observe that $m_i \coloneqq |M_i| \geq n-t$. 
We also use the following notation that allows us to rearrange the values in each coordinate and reassign the indices:

\begin{notation}[Coordinate-wise sorting]\label{obs: reordering}
For each coordinate $k\in [d]$, let $v_j[k]$ denote the $k^{th}$ coordinate of the vector received from node $j$. 
We order the values $v_j[k],\forall k\in [n]$ in increasing order and relabel the indices of the vectors accordingly. Thus, $v_j[k]$ now holds the $j^{th}$ smallest value in coordinate $k$. Note that after the renaming, for two different coordinates $k$ and $l$, $v_j[k]$ and $v_j[l]$ may not hold coordinates received from the same node.  
\end{notation}

With this notation, we can now define the locally trusted box:

\begin{definition}[Locally trusted box]
    For all received vectors $v_j\in M_i$ by node $i$, denote $v_j[k]$ the $k^{th}$ coordinate of the respective vector. For each coordinate, we reassign the indices as in \Cref{obs: reordering}.
    The number of \byzadj values for each coordinate is at most $m_i-(n-t)$. 
    The locally trusted box $\tbi$ computed by node $i$ is the Cartesian product of 
    $\bigl[v_{m_i-(n-t)+1}[k], v_{n-t}[k]\bigr] $ for $ k\in[d]$. 
    Denote $\tbi[k]$ the orthogonal projection of $\tbi$ onto the $k^{th}$ coordinate. 
\end{definition}

In other words, since in the synchronous setting each node $i$ is guaranteed to receive the $n-t$ true vectors, the number of received messages $m_i$ allows node $i$ to infer the number of $\byzadj$ vectors it received and trim this exact number of values on each side of the interval for each coordinate. This ensures that $\tbi[k]$ always contains $n-2t$ values, for every coordinate $k$. 

Observe that since we remove the maximum number of potentially \byzadj values for each coordinate on each side when computing the locally trusted box $\tbi$, any locally trusted box is contained in the trusted box $\tb$.

\section{Synchronous algorithms to compute an approximation of the centroid}

We present here synchronous algorithms for computing an approximation of the true centroid, while satisfying different validity conditions. Our main result is given in \Cref{sec:synch_sqrt_d_approximation}.

\subsection{Lower bound on the approximation of the \texorpdfstring{$\boldsymbol{\safe}$}{safe} approach}\label{sec:bad_approx_safe_area}

In this section, we prove that satisfying convex validity comes at the cost of good centroid approximation.

\begin{restatable}{theorem}{existingSafeApprox}\label{thm:sync_safe_area}
    The approximation ratio of the true centroid that can be achieved by any $\safe$ approach is at least $2d$.
\end{restatable}

\begin{proof}
In order to prove the lower bound on the approximation ratio, we present a construction where the $\safe$ consists of just one node, and the distance between this node and $\trueBar$ is at least $2d$.  
To make sure that $\trueBar$ is as far as possible from the $\safe$, we place as many correct nodes as possible outside the $\safe$.

Observe that, in $\mathbb{R}^d$, a hyperplane is a subspace of codimension $1$, i.e.\ containing the whole space except a one dimensional subspace that can be spanned by a single vector. Let this vector characterize the hyperplane.  If $d$ distinct hyperplanes are characterized by $d$ linearly independent vectors, then their intersection has codimension $d$, meaning that the intersection has dimension $0$. Hence, $d$ hyperplanes characterized by $d$ linearly independent vectors intersect in one single point at most. 

We use this property to build an example where the $\safe$ is reduced to a single point. Assume the worst case $f=t$ and that $n=(d+1)t+1$, and thus there are $d\cdot t+1$ correct vectors.
Let $v_0$ be the input vector of one correct node. The idea is to divide the rest of the correct nodes into $d$ groups of $t$ nodes with input vectors $v_i, i\in[d]$, all at a negligible distance from each other, and all linearly independent.
We then assume that the $t$ \byzadj nodes choose their input vector to be $v_0$. 
In particular, we define $v_0 = (0, 0, \dots, 0)$ and $v = (x, 0, \cdots, 0)$, where $x> 0$. We further define $d$ vectors $\delta_j = \delta\cdot u_j$ where $u_j$ is the $j^{th}$ unit vector. 
Then, we can define the position of the $d$ groups of nodes near $v$ as $v_j = v + \delta_j, \forall j \in [d]$. 
Consider the hyperplane $H_i$ spanned by all nodes except for the $t$ nodes with input vector $v_i$. 
This is, in fact, a hyperplane of $\mathbb{R}^d$, since the $(d-1)$ inputs $v_j, \forall j\neq i$ are linearly independent.
Moreover, the hyperplanes $H_i$ are all characterized by linearly independent vectors (the vector $v_i$ characterizes $H_i$), implying that their intersection must be the single point $v_0$. Thus, $\safe$ is just a single point $v_0$. See Figure~\ref{fig:lowerBoundCaseExistingSafe} for a visualization of the construction in $d=3$ dimensions. 

Since we can choose $\delta$ to be arbitrarily small, and thus ensure that the distance between any $v_j, j\neq 0$ and $v$ is negligible, we will replace coordinates $v_j$ by $v$.
We can compute the true centroid as $\trueBar = \frac{d\cdot t}{d\cdot t+1}v$, implying that 
    $\distance(\safe, \trueBar) = \frac{d\cdot t}{d\cdot t+1}x.$
Observe that $\trueBar$ and $\frac{(d-1)t}{d\cdot t+1}v + \frac{t+1}{d\cdot t+1}v_0$ are two extreme centroids defining the diameter of $\encBall$. Then, we can rewrite the diameter of the smallest enclosing ball of $\barSet$ as 
\begin{align*}
    2\cdot\radiusEncBall = \frac{d\cdot t}{d\cdot t+1}x - \frac{(d-1)t}{d\cdot t+1}x
    = \frac{t}{d\cdot t+1}x.
\end{align*}
Therefore, 
$$\frac{\distance(\safe, \trueBar)}{\radiusEncBall} = \frac{d\cdot t}{d\cdot t+1}\cdot x\cdot \frac{d\cdot t+1}{t}\cdot\frac{2}{x} = 2d,$$ which gives the desired lower bound on the approximation ratio. \qedhere
\end{proof}

\begin{restatable}{observation}{unboundedSyncSafeArea}\label{thm:unbounded_sync_safe_area}
    Assume that the barycenter of the $\safe$ from~\cite{mendes2015multidimensional} is chosen to approximate the centroid. Then the approximation ratio of the true centroid is unbounded.
\end{restatable}

\begin{proof}
    Note that $\radiusEncBall$ can be $0$. This happens when the \byzadj nodes do not send any vectors at all and only one possible centroid can be computed. 
    Let $t+1$ correct nodes have input vector $v_0$ and $n-2t-1$ correct nodes have input vector $v_1$, the $\safe$ is spanned by $[v_0,v_1]$. Observe that the nodes of the $\safe$ do not have weights. In the first round, all correct nodes will select the barycenter $(v_0+v_1)/2$ as their input vector for the next round. $\trueBar$, however, lies in $((t+1)\cdot v_0 + (n-2t-1)\cdot v_1) / (n-t)$, and $\radiusEncBall = 0$ (since only $n-t$ vectors are received, $|\barSet|=1$). The approximation ratio is thus unbounded. \qedhere
 \end{proof}
 
\begin{restatable}{corollary}{lowerBoundSyncConvex}
\label{cor: lowerBoundSyncConvex}
    The approximation ratio of the true centroid that can be achieved by any algorithm satisfying convex validity is at least 2d.
\end{restatable}

\begin{proof}
    We assume that $t$ \byz follow the algorithm with their own (worst-case) input vectors, thus being undetectable. 
    Let us consider an approximate agreement algorithm such that the convergence vector always satisfies convex validity. For the sake of contradiction, suppose this convergence vector is outside the $\safe$ computed in the very first round for a specific input. Then, by definition of $\safe$, there exists a convex hull $\mathrm{Conv}$ of $n-t$ input vectors that does not contain the convergence vector. 
    Since the \byz are undetectable, we can choose an input where the $n-t$ vectors corresponding to $\mathrm{Conv}$ are correct, and the rest are \byzadj. The convergence vector then violates the convex validity property, which results in a contradiction.
    Hence, the convergence vector of any algorithm satisfying the convex validity property must be in the $\safe$. The corollary then follows from \Cref{thm:sync_safe_area}.
\end{proof}

\begin{figure}[tbh]
    \centering
    \includegraphics[width=0.7\textwidth]{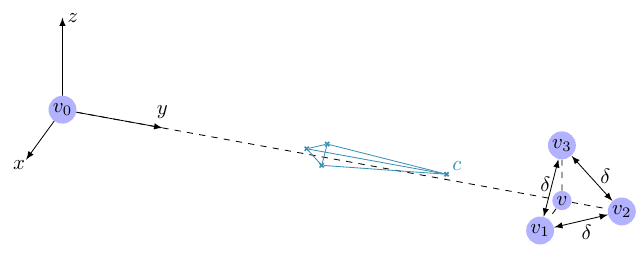}
    \caption{Lower bound construction with $d=3$ and $n=4t+1$. The input vectors are spread into $d+1=4$ groups, one at $v_0$ ($t+1$ of them), and the others at $v_i$, for $i=1, 2, 3$ (groups of $t$ vectors).  
    Here $\convexHull\bigl({v_0, v_1, v_2}\bigr)\cap \convexHull\bigl({v_0, v_1, v_3}\big)\cap \convexHull\bigl({v_0, v_2, v_3}\bigr) = \{v_0\} = \safe$. 
    The light blue area is $\convexHull(\barSet)$, where $c \coloneq \trueBar$.}
    \label{fig:lowerBoundCaseExistingSafe}
\end{figure}

\subsection{Optimal approximation with weak validity via the \texorpdfstring{$\boldsymbol{\safe}$}{safe} approach}\label{sec:synch_one_approximation}

To obtain an optimal approximation of the centroid, it is possible to use a $\safe$ approach and add a preprocessing round where nodes calculate local approximations of the centroid. Algorithm~\ref{alg:safe area approach} presents this idea in pseudocode.
In this algorithm, nodes first locally compute the set of all possible centroids of $n-t$ vectors. By \Cref{def:approx}, we know that agreeing inside $\encBall$ ensures an optimal approximation of the centroid, hence each node picks the center of the smallest enclosing ball of the computed possible centroids as their new vector. Then, each node calls the $\ExistingSafeArea$ subroutine (synchronous version of Algorithm $6$ from~\cite{mendes2015multidimensional}), which ensures that each node outputs a vector inside $\convexHull(\{c_i, i\in [n] \})$. Since each of the $c_i$ provides an optimal approximation, agreeing inside $\convexHull(\{c_i, i\in [n] \})$ preserves this property. 
There are a few subtleties to the proof since each node computes its own local set of possible centroids, which possibly differs from $\barSet$.
Additionally, we show that \Cref{alg:safe area approach} does not satisfy the strong validity condition as shown in \Cref{fig:safe-convb-not-in-tb}, but satisfies the weak validity condition. 
Also, the use of the $\ExistingSafeArea$ subroutine forces exponential local computations and a resilience of $t<\min\left(\frac{n}{3}, \frac{n}{d+1}\right)$, which gets further away from the $t<n/3$ optimum the more the dimension grows.

\begin{algorithm}[tbh]
\caption{Synchronous optimal approximation of the centroid with $t<\min(n/3, n/(d+1))$}
\label{alg:safe area approach}
\begin{algorithmic}[1]
    \State Each node $i$ in $[n]$ with input vector $v_i$ executes the following code:
    \Indent
    \State Broadcast $v_i$ reliably to all nodes
    \State Reliably receive up to $n$ messages $M_i = \{v_j, j\in [n] \}$ 
    \State Compute centroids of all subsets of $n-t$ vectors from $M_i$ and store them in $\barSet(i)$
    \State Set new vector $c_i$ to be the center of the smallest enclosing ball of $\barSet(i)$
    \State Run $\ExistingSafeArea$ with input vector $c_i$
    \EndIndent
\end{algorithmic}
\end{algorithm}

\begin{theorem}\label{thm:synch-one-approx}
    \Cref{alg:safe area approach} achieves an optimal approximation of $\ \trueBar$ in the synchronous setting while satisfying weak validity. After $O\bigl(\log_{1/(1-\gamma)}\bigl(\frac{2}{\varepsilon}\cdot\radiusEncBall\cdot\sqrt{d}\bigr)\bigr)$ synchronous rounds, where $1/(1-\gamma) = n\binom{n}{n-t}\big/\bigl(n\binom{n}{n-t}-1\bigr)$, the vectors of the correct nodes satisfy the $\varepsilon$-agreement property. The resilience of the algorithm is $t<\min(n/3, n/(d+1))$.
\end{theorem}

In the following, we prove the above theorem in two steps and discuss why the consensus vector from \Cref{alg:safe area approach} does not satisfy strong validity. 

\begin{lemma}\label{lem:safeApproachApprox}
    \Cref{alg:safe area approach} outputs an optimal approximation of $\trueBar$. 
\end{lemma}

\begin{proof}
In \Cref{alg:safe area approach}, during the first round, each correct node computes its local set of centroids $\barSet(i)\subseteq \barSet$. 
The set $\barSet$ itself cannot necessarily be computed by all correct nodes, since it is not guaranteed that they all receive $n$ messages. 
The set $\barSet(i)$ depends on the number of messages $|M_i|$ received by node $i$ during this first round. 
Using $\barSet(i)$, each node $i$ computes its local view of the smallest enclosing ball of this set $\mathrm{Ball}(\barSet(i))$, denoted by $\encBalli$.  
Note that the radius of $\encBalli$ is at most $\radiusEncBall$, since the smallest enclosing ball of a subset of vectors must be smaller than the smallest enclosing ball of the original set. 
Moreover, each locally computed smallest enclosing ball contains $\trueBar$, i.e., 
for all correct nodes $i$, $\trueBar\in\convexHull(\barSet(i))\subseteq \mathrm{Ball}(\barSet(i))$. This holds because $M_i$ always contains all vectors from $n-f$ correct nodes (due to synchronous communication). Further, we can use the same argument as in \Cref{obs: truebar in enc ball}, since $\trueBar$ can be written as a convex combination of elements of $\barSet(i)$.
After the first round of \Cref{alg:safe area approach}, each correct node $i$ therefore has a new input vector $c_i$ at a distance of at most $\radiusEncBall$ away from $\trueBar$. 

After the first round, the subroutine $\ExistingSafeArea$ ensures that the final agreement vectors will be inside the convex hull of vectors $c_i$ that were computed in the first round. We denote this convex hull $\mathrm{Conv}_c$. 
\Cref{alg:safe area approach} therefore agrees inside $\mathrm{Conv}_c \coloneqq \convexHull\left(\{c_i\mid \forall i\in [n] \text{ correct}\}\right)$.

Recall that every $c_i$ is at a distance of at most $\radiusEncBall$ away from $\trueBar$, that is, for all correct $i\in [n]$, we get $$\distance(c_i, \trueBar)\leq \radiusEncBall.$$ 
Since $\mathrm{Conv}_c$ is a convex polytope, the distance to all vectors inside their convex hull can also be upper bounded by $$\distance(x, \trueBar)\leq \radiusEncBall \quad \forall x\in \mathrm{Conv}_c.$$
Therefore, the approximation ratio of \Cref{alg:safe area approach} is 
$$\frac{\max_{x\in \mathrm{Conv}_c}(x, \trueBar)}{\radiusEncBall} = 1.$$ 

\end{proof}

It remains to argue about the convergence of the algorithm. The $\ExistingSafeArea$ algorithm has a convergence ratio that depends on the maximum distance between any two true vectors. In the following, we show the convergence rate of \Cref{alg:safe area approach}.

\begin{lemma}\label{lem:safeApproachRounds}
    \Cref{alg:safe area approach} converges in $O\left(\log_{1/(1-\gamma)}\left(\frac{2}{\varepsilon}\cdot\radiusEncBall \cdot\sqrt{d}\right)\right)$ rounds, where $1/(1-\gamma) = n\binom{n}{n-t}\Big/\left(n\binom{n}{n-t}-1\right)$, while satisfying weak validity.
\end{lemma}
\begin{proof}
Recall that, after one single round, all correct nodes have vectors $c_i$ that are inside $\encBall$, indexing correct nodes with $i\in[n-f]$. 
Hence, after this first round,\\
${\max_{\substack{k\in [d]\\i, j\in[n-f]}}|c_i[k]-c_j[k]|\leq 2\radiusEncBall}$. 
By~\cite{mendes2015multidimensional} and by using $f\le t$, $\ExistingSafeArea$ therefore converges in \\$O\left(\log_{1/(1-\gamma)}\left(\frac{2}{\varepsilon}\cdot\radiusEncBall \cdot\sqrt{d}\right)\right)$ rounds. 

Finally, we show that this algorithm satisfies weak validity. We assume that all nodes are correct and start with the same input vector. Then, in the first round, every node will receive a set of $n$ identical vectors. Every node will thus have the same input to $\ExistingSafeArea$, thus satisfying weak validity.
\end{proof}

This Lemma concludes the proof of \Cref{thm:synch-one-approx}. In the following, we discuss that the consensus vector from \Cref{alg:safe area approach} does not satisfy strong validity. 

\begin{observation} 

    The example in \Cref{fig:safe-convb-not-in-tb} shows that \Cref{alg:safe area approach} does not satisfy the strong validity condition from \Cref{def:multidim_agreement_properties}. Therefore, the \byzadj nodes can make the correct nodes agree on almost any vector. 
\end{observation}

\begin{figure}[tbh]
    \centering
    \includegraphics[width=0.7\textwidth]{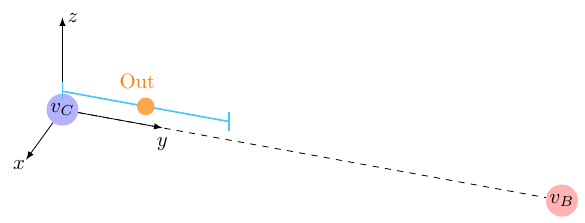}
    \caption{In the above example, suppose all correct nodes start with vector $v_C$ and all \byzadj nodes start with vector $v_B$. The segment in blue is the diameter of $\encBall$ and the orange node $\mathrm{Out}$ is the output of every correct node if all of them receive all $n$ vectors. However, to satisfy the strong validity condition from \Cref{def:multidim_agreement_properties}, the output of every correct node should be $v_C$. }
    \label{fig:safe-convb-not-in-tb}
\end{figure}

\subsection{Constant approximation of the centroid with strong validity via the \texorpdfstring{$\boldsymbol{\mda}$}{mda} approach}\label{sec:synch_const_approximation}

We analyze here the minimum-diameter averaging ($\mda$)~\cite{pmlr-v80-mhamdi18a} approach, as described in~\cite{NEURIPS2021_d2cd33e9}. We show it can be used to derive a constant approximation of the centroid while satisfying strong validity, and that it doesn't satisfy the box validity condition. 
In the following, we formally define $\mda$.

\begin{definition}[$\mda(M,m)$]
    Let $M = \{v_i, i\in [n] \}$ be a set of at least $m$ vectors. Then 
    $$\mda(M,m) = \argmin_{M'\subset M;\ |M'|=m} \max_{u,v\in M'} \distance(u,v)$$
\end{definition}

Algorithm~\ref{alg:mda approach} presents the $\mda$ approach as pseudocode. The idea is to let nodes repeatedly compute the $\mda$ of a subset of $n-t$ of the received vectors. Each node then locally chooses the centroid of the $\mda$ as its new vector. In particular, in the first round, each node locally picks one vector from $\barSet$. This property allows us to derive the constant approximation ratio, since every element of $\barSet$ is a $2$-approximation of $\trueBar$. Let $V_{\text{correct}}$ denote the set of input vectors of the correct nodes, we prove the following. 

\begin{theorem}\label{thm:mda_correctness}
    \Cref{alg:mda approach} converges in $O\Bigl(\log_2 \bigl(\frac{1}{\varepsilon}\cdot D\bigr)\Bigr)$ rounds, where $D$ is the diameter of the correct vectors (i.e., $D= \max_{u,v\in V_{\text{correct}}} \distance(u,v)$), and achieves a $3.8$-approximation of $\ \trueBar$ in the synchronous setting. The vectors of the correct nodes satisfy strong validity, and the resilience of the algorithm is $t<n/4$.
\end{theorem}

\begin{algorithm}[tbh]
\caption{Synchronous constant approximation of the centroid with $t<n/4$}
\label{alg:mda approach}
\begin{algorithmic}[1]
\State In each round $r=1,2,\ldots, 6\cdot\left\lceil\log_2 \bigl(\frac{1}{\varepsilon}\cdot D\bigr)\right\rceil$, where $D$ is the diameter of the vectors and $C>1$ a large constant, each node $i$ in $[n]$ with input vector $v_i$ does the following:
    \Indent
    \State Broadcast $v_i$ reliably to all nodes
    \State Reliably receive up to $n$ messages $M_i = \{v_j, j\in [n] \}$ 
    \State Compute $MDA(M_i, n-t)$ 
    \State Set new vector $v_i$ to be the centroid of $MDA(M_i, n-t)$\label{alg:mda approach:line:decision}
    \EndIndent
\end{algorithmic}
\end{algorithm}

The $\mda$ algorithm has a good (but not optimal) resilience of $t<n/4$. 
However, it has two disadvantages: the computation of $\mda$ requires local exponential runtime in $f$, similar to the computation of $\safe$; and the strong validity condition is the best validity condition we can satisfy with this approach (\Cref{lem:mda_no_box}).  

\begin{observation}\label{obs:synch_mda_and_safe}
    There exists a $2$-approximation of the centroid computed using a combination of $\mda$ and the $\safe$ approach that satisfies strong validity and $t<\min(n/4, n/(d+1))$.
\end{observation}

Combining those two approaches gives a tight approximation ratio, as we prove below that any algorithm satisfying the strong validity property can achieve at best a $2$-approximation of $\trueBar$.

\begin{lemma}\label{lem:MDAsynch}
    Algorithm~\ref{alg:mda approach} converges in $O\Bigl(\log_2 \bigl(\frac{1}{\varepsilon}\cdot D\bigr)\Bigr)$ rounds, where $D\coloneqq \max\limits_{u,v\in V} \distance(u,v)$ is the diameter of the true vectors, and achieves a $3.8$-approximation of the centroid if $t<n/4$.
\end{lemma}
 
\begin{proof}
    We start by showing convergence. Due to synchronous communication, all true vectors will be received by all correct nodes in every round. When computing $\mda$, Byzantine vectors may be taken into account, but only if they are not further away than the diameter of the correct nodes away from the correct nodes. Let $D_r$ be the diameter of the true input vectors of round $r$. Consider two vectors $u_i$ and $u_j$ that are centroids of $\mda(M_i, n-t)$ and $\mda(M_j, n-t)$ computed in round $r$ in line~\ref{alg:mda approach:line:decision} of Algorithm~\ref{alg:mda approach} by two correct nodes $i$ and $j$, respectively. Note that $\mda(M_i, n-t)\cap \mda(M_j, n-t)$ consists of $k \ge n-2t$ vectors. Let $w_i\in \mda(M_i, n-t)\setminus \mda(M_j, n-t)$ and $w_j\in \mda(M_j, n-t)\setminus \mda(M_i, n-t)$ be the two vectors in the respective sets that have the largest distance from each other. Note that these vectors can be at a distance of at most $2D_r$ from each other, as the sets may have been formed with Byzantine vectors. Then, the distance between $u_i$ and $u_j$ can be upper bounded as
    \begin{align*}
    \distance(u_i,u_j) &\le \frac{1}{n-t}\distance((|M_i|-k)\cdot w_i,(|M_j|-k)\cdot w_j)\le \frac{1}{n-t}\distance(t\cdot w_i,t\cdot w_j) \\
    &\le \frac{t}{n-t}\distance(w_i, w_j) \le \frac{2t}{n-t}D_r < \frac{2}{3}D_r,
    \end{align*}
    since $t<n/4$. 
    This inequality holds for any pair of nodes $i$ and $j$. Thus, $D_{r+1} < 2D_r / 3$ holds for all rounds $r$. Let $\varepsilon$ be a given small constant. Algorithm~\ref{alg:mda approach} terminates after $6\cdot\left\lceil\log_2 \bigl( D/\varepsilon\bigr)\right\rceil$ rounds. The corresponding diameter of the true vectors is then $D\cdot(2/3)^{6\cdot\lceil\log_2 (D/\varepsilon)\rceil} \le\varepsilon$. Thus, Algorithm~\ref{alg:mda approach} satisfies $\varepsilon$-agreement. 

    We next focus on the approximation statement. Our goal is to show that     
    \begin{align*}
        \sum_{r=1}^\infty\distance(v_i^{(r-1)}, v_i^{(r)})\leq  3.8\cdot\radiusEncBall
    \end{align*}
    where $v_i^{(r)}$ is vector of a correct node $i$ at the end of round $r$. 
    Since the approximation ratio is computed as $\distance(v_i^{\lceil\log_2 ( D/\varepsilon)\rceil}, \trueBar)/\radiusEncBall$, by using the triangle inequality, the approximation follows. 

    Note that $v_i^{(0)}$ is the input vector of node $i$. 
    Let $V_i^{(r)}$ denote the set of vectors that will be received by node $i$ at the beginning of round $r+1$, and $V_{\mathrm{true}}^{(r)}$ denote the set of true vectors at the beginning of round $r+1$. 
    Note that, due to synchronous communication, the set $V_i^{(r)}$ consists of $V_{\mathrm{true}}^{(r)}$ and the vectors that are sent by $\byzadj$ nodes to node $i$ at the beginning of round $r+1$. Hence, $|V_{\mathrm{true}}^{(r)}| = n-t$ and $n-t\leq |V_i^{(r)}|\leq n$.
    We further use $\mda_i^{(r)} = \argmin_{M\subset V_i^{(r-1)};\ |M|=n-t} \max_{u,v\in M} \distance(u,v)$ to denote the $\mda$ chosen by node $i$ during round $r$, and $\text{diam}(\mda_i^{(r)}) = \max_{u,v\in \mda_i^{(r)}} \distance(u,v)$.

    Observe that the input vector for the second round of the algorithm $v_i^{(1)}$ is computed by taking the centroid of $n-t$ received vectors. Any such centroid computed from $n-t$ vectors is, by definition, in $\encBall$ and thus a $2$-approximation of $\trueBar$. Any centroid can be at a distance of at most $2\radiusEncBall$ away from $\trueBar$. Hence, after the first round of the algorithm, $\max_{u,v\in V_{\mathrm{true}}^{(1)}}\distance(u,v)\leq 2\radiusEncBall$.

    From the convergence, we have 
        $\max_{u,v\in V_{\mathrm{true}}^{(r+1)}} \distance(u,v) \leq \frac{2t}{n-t}\max_{u,v\in V_{\mathrm{true}}^{(r)}} \distance(u,v), \forall r\geq 1.$
    Further, observe that, either $\mda_i^{(r)} = V_{\mathrm{true}}^{(r-1)}$, or $\text{diam}(\mda_i^{(r)})\leq \max_{u,v\in V_{\mathrm{true}}^{(r-1)}} \distance(u,v)$ since $V_{\mathrm{true}}^{(r-1)} \subseteq V_i^{(r-1)}$.

    Now, we can bound $\distance(v_i^{(r-1)}, v_i^{(r)})$ using $\max_{u,v\in V_{\mathrm{true}}^{(r-1)}}\distance(u, v)$. 
    Vector $v_i^{(r-1)}$ is in the set of vectors that will be considered for the computation of $\mda_i^{(r)}$ during round $r$. 
    If $\mda_i^{(r)} = V_{\mathrm{true}}^{(r-1)}$, then the new vector $v_i^{(r)}$ computed by node $i$ will be the centroid of vectors in $\mda_i^{(r)}$ and hence be at most $\text{diam}(\mda_i^{(r)}) = \max_{u,v\in V_{\mathrm{true}}^{(r-1)}}\distance(u, v)$ away from $v_i^{(r-1)}$, since $v_i^{(r-1)}\in \mda_i^{(r)}$. 
    Otherwise, $\mda_i^{(r)}$ contains at most $t$ $\byzadj$ vectors, and at least $n-2t$ vectors from $V_ {\mathrm{true}}^{(r-1)}$. 
    The new vector computed by node $i$ is still the centroid of vectors of $\mda_i^{(r)}$, however $v_i^{(r-1)}\notin \mda_i^{(r)}$. 
    Hence, since $t<n/4$, 
    \begin{align*}
        \distance(v_i^{(r-1)}, v_i^{(r)})&\leq 
        \max_{u,v\in V_{\mathrm{true}}^{(r-1)}}\distance(u, v) \\
        &\qquad\qquad\qquad+ \distance \left(\centroidFunction\left(V_{\mathrm{true}}^{(r-1)}\cap \mda_i^{(r)}\right), \centroidFunction(\mda_i^{(r)})\right)\\
        &\leq \left(1+\frac{t}{n-t}\right)\max_{u,v\in V_{\mathrm{true}}^{(r-1)}}\distance(u, v)
    \end{align*}
    for $r\ge2$. Finally, 
    \begin{align*}
        \sum_{r=1}^\infty\distance(v_i^{(r-1)}, v_i^{(r)})
        &\leq 2\radiusEncBall + \sum_{r=2}^\infty \left(1+\frac{t}{n-t}\right)\max_{u,v\in V_{\mathrm{true}}^{(r-1)}}\distance(u, v)\\
        &\leq 2\radiusEncBall + \left(1+\frac{t}{n-t}\right)  \sum_{r=2}^\infty \max_{u,v\in V_{\mathrm{true}}^{(r-1)}}\distance(u, v)\\
        & \leq 2\radiusEncBall + \frac{n}{n-t}  \sum_{r=2}^\infty \left(\frac{2t}{n-t}\right)^{r}\max_{u,v\in V_{\mathrm{true}}^{(1)}}\distance(u, v)\\
        &\leq 2\radiusEncBall +\frac{n}{n-t} \sum_{r=2}^\infty \left(\frac{2t}{n-t}\right)^{r}\cdot 2\radiusEncBall \\
        &\leq \left( 2+\frac{4t^2n}{(n-t)^2(n-3t)}\right)\radiusEncBall \\
        &<\frac{34}{9}\cdot\radiusEncBall<3.8\cdot\radiusEncBall.
    \end{align*}

    Observe that this bound tends to $2$ as $t$ tends to $0$, meaning that a better approximation can be achieved with fewer Byzantine nodes.

\end{proof}

\begin{restatable}{lemma}{mdaNoBox}\label{lem:mda_no_box}
    \Cref{alg:mda approach} does not satisfy box validity.
\end{restatable}

\begin{proof}
    Let $v_0=(0,\dots,0)$ and $v_1=(1,0\ldots,0)$ be the possible input vectors of the correct nodes, where $n-2t$ nodes hold the vector $v_0$ and $t$ nodes hold $v_1$. Box validity is satisfied if the final vectors lie on the line between $v_0$ and $v_1$. Let the $t$ \byz parties now pick a vector at $(0,0.999, 0,\dots, 0)$ as its input. Note that the \byz vectors will be chosen together with vectors at $v_0$ by the $\mda$. The centroid of these vectors lies outside of the box of the correct input vectors. By repeating this procedure over several rounds, the \byz party can make sure that the correct nodes converge outside of the box.  
\end{proof}

\begin{restatable}{lemma}{lowerBoundSyncStrongVal}\label{lem:lowerBound-sync-strongVal}
    In the synchronous setting, the approximation ratio of the true centroid that can be achieved by any algorithm satisfying the strong validity property is at least 2.
\end{restatable}

\begin{proof}
    Consider the setting where $n-t$ input vectors are $0$ and $t$ input vectors have $0$ in all coordinates but $1$ in the first one, i.e.\ $e_1 = \left(1, 0, \dots, 0\right)$. 
    Consider an algorithm that satisfies the strong validity condition. Then, its convergence vector has to be $0$. Indeed, the \byz could have the $t$ input vectors  $e_1$, in which case all true vectors would have the same input vector $0$ and the algorithm would have to converge to it by definition of the strong validity condition.  

    However, the \byz can all have input vectors at $0$, in which case the true centroid is $\left(\frac{n-2t}{n-t}\cdot 0 + \frac{t}{n-t}\cdot 1, 0, \dots, 0\right) = \left(\frac{t}{n-t}, 0, \dots, 0\right)$, which is at distance $\frac{t}{n-t}$ of the true centroid $0$. 
    Since $0$ and $\left(\frac{t}{n-t}, 0, \dots, 0\right)$ are the two furthest elements of $\barSet$, $\radiusEncBall = \frac{1}{2}\cdot \frac{t}{n-t}$. Hence, the approximation ratio in this specific case is $2$.
\end{proof}

\subsection{A \texorpdfstring{$\mathbf{2\sqrt{d}}$}{2sqrtd}-approximation of the centroid with box validity}\label{sec:synch_sqrt_d_approximation}

In this section, we relax the convex validity from~\cite{mendes2015multidimensional} to the box validity condition. Though this condition is looser than the convex validity, it is still way more restrictive than the strong validity condition. This allows us to improve the approximation ratio as well as the resilience of the $\safe$ approach from Section~\ref{sec:bad_approx_safe_area}, but also allows us to satisfy a better validity condition than the ones satisfied by algorithms in \Cref{sec:synch_one_approximation} and \Cref{sec:synch_const_approximation}. This is shown by \Cref{thm:syncBoxAlg}. 

\begin{figure}[tbh]
    \centering
    \includegraphics[width=0.8\textwidth]{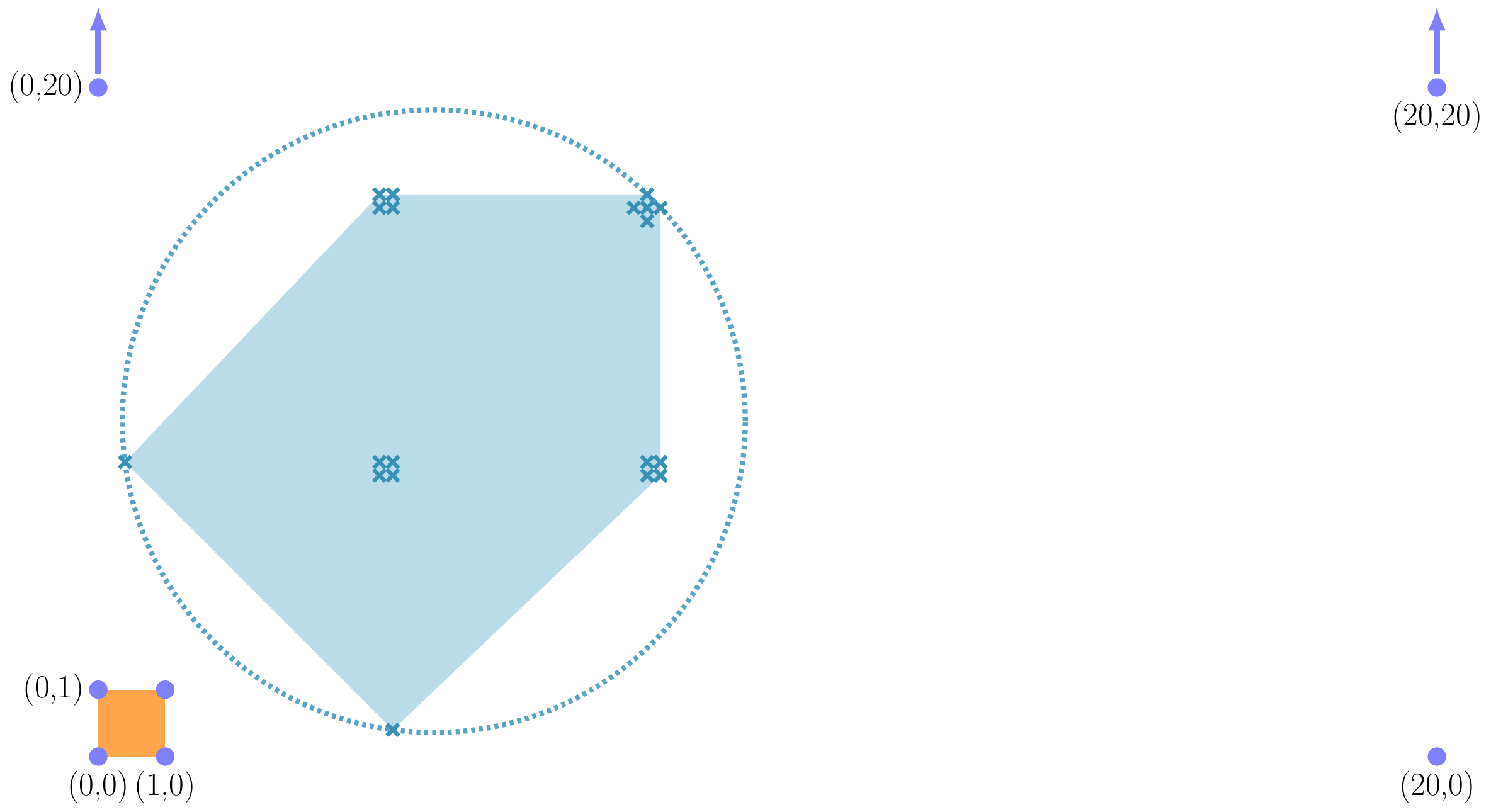}
    \caption{An example of initial configuration for the $2$-dimensional agreement problem. Solid circles represent the $n=7$ input vectors, including up to $t=2$ \byzadj vectors. The crosses represent all possible centroid positions, with the blue area showing their convex hull.  
    The orange box will be contained in $\tbi$ for all $i\in [n]$, and contains the $\safe$.  
    The blue circle is the area where an optimal approximation of the centroid can be achieved.  
    }
    \label{fig:intro-example}
\end{figure}

In Figure~\ref{fig:intro-example}, we show an example where neither $\barSet$, nor the smallest enclosing ball of $\barSet$ intersect $\tbi$. Thus, no node can locally compute an optimal approximation of $\trueBar$. Note that the smallest box around $\barSet$ would intersect $\tbi$ in the same example. We formally define this box as follows:

\begin{definition}[Centroid box]
    The centroid box $\bb$ is the smallest box containing $\barSet$, and the local centroid box $\bbi$ of node $i \in [n]$ is the smallest box containing $\barSet(i)$. 
\end{definition}

In the following, we make the observation that the local centroid box can be indeed computed in polynomial time locally:

\begin{observation}
$\bbi$ can be computed in polynomial time as follows: for each coordinate, compute the average of the $n-t$ smallest values and the average of the $n-t$ largest values. This yields the interval defining the box over this coordinate.     
\end{observation}

Next, we define the midpoint function that is used in the algorithm.

\begin{definition}[Midpoint]
    The $\midpointFunction$ of a box $X$ is defined as 
$$\midpointFunction(X) = \Bigl(\midpointFunction\bigl(X[1]\bigr), \dots, \midpointFunction\bigl(X[d]\bigr)\Bigr), $$
where $X[k]$ is the set containing all $k^{th}$ coordinates of vectors of the set $X$, and the one-dimensional $\midpointFunction$ function returns the midpoint of the interval spanned by a finite multiset of real values.
\end{definition}

\Cref{alg:box approach} solves multidimensional approximate agreement with box validity. The idea of this algorithm is to compute a box containing trusted vectors and a box containing all possible centroids. In every round, the new input vector is set to be the midpoint of the intersection of the two boxes. The presented algorithm has optimal resilience of $t<n/3$ and reaches a $2\sqrt{d}$ approximation of the true centroid, while using polynomial local computation time. 
Thus, the algorithm outperforms the $\safe$ approach in resilience, consensus quality, and computation time. The termination round is defined by each node locally in the first round of the algorithm. If some node terminates earlier than others, the other nodes can use its last sent vector for all following rounds.
In order to analyze the convergence of the algorithm, we will define the convergence rate with respect to the longest edges of the boxes:

\begin{definition}[Length of the longest edge (\longestEdge)]
     The length of the longest edge of a box $B$ is defined as
         $$\longestEdge(B) = \max_{\substack{k\in[d]\\ u, v\in B}} \bigl|u[k] - v[k]\bigr|.$$    
 \end{definition}

Note that $\longestEdge(\tb) = \max_{k\in[d]; i, j\in I_{\text{correct}}} \bigl|v_i[k] - v_j[k]\bigr|,$ where $I_{\text{correct}}$ denotes the set of indices of the true vectors, and the longest edge thus corresponds to the diameter of the true vectors.
The following theorem shows the main properties of the algorithm.

\begin{algorithm}[tbh]
\caption{Synchronous approximate agreement with box validity and resilience $t<n/3$}
\label{alg:box approach}
\begin{algorithmic}[1]
    \State Let $C$ be a large constant. Each node $i \in [n]$ with input vector $v_i$ executes the following code in each round $r=1,2,\ldots,C\cdot\log_2 \bigl(\frac{1}{\varepsilon}\cdot\longestEdge\left(\tbi\right)\bigr)$:
        \Indent
            \State Broadcast $v_i$ reliably to all nodes
            \State Reliably receive up to $n$ messages $M_i = \{v_j, j\in [n]\}$
            \State Compute $\tbi$ from $M_i$ by excluding $|M_i|-(n-t)$ values on each side 
            \State Compute $\bbi$ from $M_i$ 
            \State Set $v_i$ to $\midpointFunction(\tbi\cap \bbi)$.\label{algline: intersection-tbi-bbi}
        \EndIndent
\end{algorithmic}
\end{algorithm}

\begin{theorem}\label{thm:syncBoxAlg}
    \Cref{alg:box approach} achieves a $2\sqrt{d}$-approximation of $\trueBar$ in the synchronous setting. After $O\Bigl(\log_2 \bigl(\frac{1}{\varepsilon}\cdot\longestEdge\left(\tb\right)\bigr)\Bigr)$ synchronous rounds, the vectors of the correct nodes satisfy multidimensional approximate agreement with box validity. The resilience of the algorithm is $t<n/3$. 
\end{theorem}

In the following, we prove the correctness of \Cref{thm:syncBoxAlg} in several steps. 
We start by proving that \Cref{alg:box approach} is well-defined, by showing that $\tbi$ and $\bbi$ intersect if $t<n/2$. 
This way, all correct nodes are able to choose a new input vector in line~\ref{algline: intersection-tbi-bbi} of \Cref{alg:box approach}.

\begin{lemma}\label{lem:bb-and-tb-intersect}
    For all correct nodes $i$, the local trusted box and the local centroid box have a non-empty intersection: $\bbi\cap\tbi\neq \emptyset$.
\end{lemma}
\begin{proof}
Let $\tbi[k]$ and $\bbi[k]$ be the (coordinate-wise) projections of $\tbi$ and $\bbi$ on the $k^{th}$ coordinate, respectively.
We prove the statement of the lemma for some correct node $i$ (note that this proof works for every correct node).

For each coordinate $k\in [d]$, let $v_j[k]$ denote the $k^{th}$ coordinate of the vector received from node $j$. 
We sort the values $v_j[k],\forall k\in [n]$ as in \Cref{obs: reordering}.
Recall that node $i$ receives $m_i\geq n-t$ vectors and that there are $m_i$ values in each coordinate. Because we remove $m_i-(n-t)$ on each side of the interval for each coordinate $k$, each $\tbi[k]$ contains $2(n-t)-m_i\geq n-2t$ values. Our goal is to construct a point that is both in $\tbi$ and $\bbi$. 

Consider the averaged sum $c_\ell$ of the $n-t$ smallest received values in coordinate $k$: $c_\ell=\frac{1}{n-t}\sum_{j=1}^{n-t}v_j[k]$. 
Observe that there exists a subset of $n-t$ vectors with this averaged sum in their $k^{th}$ coordinate. The corresponding vector is the centroid of the $n-t$ nodes with these input vectors. It is contained in $\bbi$ and $c_\ell$ is its orthogonal projection on coordinate $k$. 
Similarly, there exists an averaged sum $c_u = \frac{1}{n-t}\sum_{j=m_i-(n-t)+1}^{m_i}v_j[k]\in \barSet[k]$ formed by the $n-t$ largest values in coordinate $k$, that corresponds to the $k^{th}$ coordinate of some centroid in $\bbi$. This definition makes sure that $[c_\ell, c_u]\subseteq \bbi[k]$.

Let $V_k = \frac{1}{2(n-t)-m_i}\sum_{j=m_i-(n-t)+1}^{n-t}v_j[k]$ be the average of values, where we truncate the $m_i-(n-t)$ smallest and largest values in this coordinate. This is possible since $t<n/3$ and $m_i\geq n-t$. Since the values $v_j[k]$ are ordered, $\frac{1}{2(n-t)-m_i}\sum_{j=m_i-(n-t)+1}^{n-t}v_j[k] \geq c_\ell$ and $\frac{1}{2(n-t)-m_i}\sum_{j=m_i-(n-t)+1}^{n-t}v_j[k] \leq c_u$. 
Therefore,  $V_k  \in [c_\ell, c_u]\subseteq \bbi[k]$.  
Also, $V_k \in \bigl[v_{m_i-(n-t)+1}[k], v_{n-t}[k]\bigr] = \tbi[k]$, since it is the mean of values $v_{m_i-(n-t)+1}[k]$ to $v_{n-t}[k]$.

Now, consider a vector $V$ s.t.\ $\forall k\in [d]$: $V[k] = V_k$. 
By the definitions of $\tbi$ and $\bbi$, since $V[k] \in \tbi[k]\cap \bbi[k], \forall k\in[d]$, this vector will be inside the intersection of the boxes, i.e.\ $V\in\tbi\cap \bbi$. \qedhere
\end{proof}

\begin{lemma}\label{lem: boxAlgCorrect}
    \Cref{alg:box approach} converges.
\end{lemma}
\begin{proof}

We start by defining $\tb$ and $\tbi$ with respect to rounds $r$ of the algorithm. 
The initial $\tb = \tb^1$ is the smallest box containing all true vectors. Let $\tb^{r+1}$ be the smallest box that contains all the vectors computed by correct nodes in round $r$, which will represent the input in round $r+1$. 
Observe that during a round $r$, each node $i$ computes $\tbi^r\subseteq\tb^{r}$, and picks its new vector inside this box. 

In order to prove convergence, we need to show that the vectors computed in each round are getting closer together.
In particular, we show that for any two correct nodes $i$ and $j$, and the corresponding local boxes $\tbi^r$ and $\bbi^r$ computed in some round $r\geq 1$, the following inequality is satisfied with a small constant $\delta<1$:
$$\bigl|\midpointFunction(\tbi^r\cap\bbi^r) - \midpointFunction(\tb_j^r\cap \bb_j^r)\bigr|\leq \delta\cdot \longestEdge(\tb^{r}).$$
We will show this inequality for $\delta = 1/2$. 

We first show that $\tbi^r\subseteq \tb^{r}$: 
At the beginning of round $r$, $\tb^{r}$ is the smallest box containing all true vectors. Consider the boxes with respect to one coordinate $k\in [d]$. 
If the interval $\tbi^r[k]$ does not contain any \byzadj values, then $\tbi^r[k]\subseteq \tb^{r}[k]$ by definition. This could have either happened if the node only received true vectors, or if the \byzadj values were removed as extreme values when computing the interval $\tbi^r[k]$.
Otherwise, we assume that the interval $\tbi^r[k]$ contains at least one \byzadj value. This implies that, when the $m_i-(n-t)$ smallest and largest values were removed, this $\byzadj$ value remained inside the interval $\tbi^r[k]$, and at least two true values were removed instead (one on each side of the interval) since there are at most $m_i-(n-t)$ \byzadj values in $M_i[k]$. Therefore, $\tbi^r[k]$ is included in an interval bounded by two true values, which is by definition included in $\tb^{r}[k]$.
Since this holds for each coordinate $k$, it also holds that $\tbi^r\subseteq \tb^{r}$. 

Next, we use \Cref{obs: reordering} to sort the received values for each coordinate $k$. 
For node $i$, we define the locally trusted box in coordinate $k$ as $\tbi^r[k] = \bigl[v_{m_i-(n-t)+1}[k], v_{n-t}[k]\bigr]$, and for node $j$ as $\tb_j^r[k] = \bigl[v'_{m_i-(n-t)+1}[k], v'_{n-t}[k]\bigr]$. Let $t_\ell$ denote the smallest and $t_u$ the largest true value in coordinate $k$ (given all true vectors). 
Since $\tbi^r[k]\subseteq\tb^{r}[k]$, $v_{m_i-(n-t)+1}[k]\geq t_\ell$ and $v_{n-t}[k]\leq t_u$. 
Similarly, $v'_{m_i-(n-t)+1}[k]\geq t_\ell$ and $v'_{n-t}[k]\leq t_u$. 
Moreover, since \byz could be inside $\tb$, $\tbi[k]$ and $\tb_j[k]$ could have been computed by removing up to $m_i-(n-t)$ true values on each side. 
Let $t_{1}$ denote the $(m_i-(n-t)+1)^{th}$ true value and $t_{2}$ the $(2n-f-t-m_i)^{th}$ true value in coordinate $k$ (recall there are $n-f$ true values in total), these true values are necessarily inside both $\tbi[k]$ and $\tb_j[k]$. 
Then, $v_{m_i-(n-t)+1}[k]\leq t_{1}$ and $v_{n-t}[k] \geq t_{2}$. 
Similarly, $v'_{m_i-(n-t)+1}[k]\leq t_{1}$ and $v'_{n-t}[k] \geq t_{2}$. 
We can now upper bound the distance between the computed midpoints of the nodes $i$ and $j$: 
\begin{align*}
    \Bigl|\midpointFunction\bigl(\tbi^r[k]\bigr) - \midpointFunction\bigl(\tb_j^r[k]\bigr)\Bigr|
    \leq \frac{t_u-t_{1}}{2} - \frac{t_{2} - t_\ell}{2} \leq \frac{t_u-t_\ell}{2}
    \leq \frac{\longestEdge(\tb^{r})}{2} .
\end{align*}
This inequality holds for every pair of nodes $i$ and $j$ and thus, for each coordinate $k$, we get
\begin{align*}
    &\max_{i, j\in [n]}\bigl|\midpointFunction\bigl(\tbi^r[k]\bigr) - \midpointFunction\bigl(\tb_j^r[k]\bigr)\bigr| \leq  \longestEdge(\tb^{r})/2\\
    & \Leftrightarrow\ \  \longestEdge\bigl(\tb^{r+1}[k]\bigr) \leq  \longestEdge(\tb^{r})/2.
\end{align*}

After $R$ rounds, $\longestEdge(\tb^R)\leq \frac{1}{2^R}\cdot \longestEdge(\tb)$ holds. 
Since there exists $R\in \mathbb{N}$ s.t.\ $\frac{1}{2^R}\cdot \longestEdge(\tb)\leq \varepsilon$, the algorithm converges.\qedhere
\end{proof}

\begin{lemma}\label{lem:syncBoxAlgRounds}
    After $\bigl\lceil \log_2\bigl(\frac{1}{\varepsilon}\cdot\longestEdge(\tb)\bigr) \bigr\rceil$ synchronous rounds of \Cref{alg:box approach}, the correct nodes hold vectors that are at a distance $\varepsilon$ from each other. 
\end{lemma}

The proof for this lemma follows from the proof of \Cref{lem: boxAlgCorrect}.
Observe that, with \Cref{lem: boxAlgCorrect} and \Cref{lem:syncBoxAlgRounds}, \Cref{alg:box approach} can be adjusted such that the nodes terminate after $\bigl\lceil \log_2\bigl(\frac{1}{\varepsilon}\cdot\longestEdge(\tb)\bigr)\bigr\rceil$ synchronous rounds. This can be achieved by letting the nodes mark their last sent message, such that other nodes can use this last vector for future computations. This way, the algorithm satisfies $\varepsilon$-agreement and termination. The following lemma shows that the algorithm satisfies box validity.

\begin{restatable}{lemma}{boxContainmentLemma}\label{lem:boxContainment}
    \Cref{alg:box approach} satisfies box validity. 
\end{restatable}

\begin{proof}
Observe that each correct node $i$ always picks a new vector inside $\tbi$. 
The set of locally received vectors depends on the received messages $M_i$. In the proof of \Cref{lem: boxAlgCorrect}, we showed that  $\tbi^r\subseteq \tb^{r}$ for all rounds $r$. This also means that $\tbi^r\subseteq \tb^1 = \tb$ always holds.
Hence, $\tbi^r\subseteq \tb$ holds for all $i$ and $r$.
This implies that the vectors converge towards a vector inside $\tb$.
\end{proof}

Observe that the algorithm also trivially satisfies the strong validity condition since it agrees inside $\tb$, and therefore solves multidimensional approximate agreement. Then, we show that \Cref{alg:box approach} provides a $2\sqrt{d}$-approximation of the true centroid.
 
\begin{lemma}\label{lem:syncBoxAlgApprox}
    The approximation ratio of the true centroid by \Cref{alg:box approach} is upper bounded by $2\sqrt{d}$. 
\end{lemma}
\begin{proof}
Following \Cref{def:approx}, we need to show that the ratio between the distance from any output of \Cref{alg:box approach} to $\trueBar$ and the optimal radius $\radiusEncBall$ is less than $2\sqrt{d}$. 
We first lower bound the radius of the smallest enclosing ball of $\barSet$. Observe that the radius is always at least  $\max_{x, y\in \barSet}\bigl(\distance(x, y)\bigr)/2$.

Next, we focus on upper bounding the distance between $\trueBar$ and the furthest possible point from it inside $\bb$. W.l.o.g., we assume that $\max_{x, y\in \barSet}\bigl(\distance(x, y)\bigr) = 1$.
We consider the relation between the smallest enclosing ball and $\bb$.
Observe that each face of the box $\bb$ has to contain at least one point of $\barSet$. 
If $\bb$ is contained inside the ball, i.e.\ if the vertices of the box would lie on the ball surface, the computed approximation of \Cref{alg:box approach} would always be optimal. 

The worst case is achieved if the ball is (partly) contained inside $\bb$. Then, the optimal solution might lie inside $\bb$ and the ball, while the furthest node may lie on one of the vertices of $\bb$ outside of the ball. The distance of any node from $\trueBar$ in this case is upper bounded by the diagonal of the box. Since the longest distance between any two points was assumed to be $1$, the box is contained in a unit cube. Thus, the largest distance between two points of $\bb$ is at most $\sqrt{d}$. 

The approximation ratio of \Cref{alg:box approach} can be upper bounded by:
$$\hspace{1.4cm}\frac{\max_{x\in \bb}\bigl(\distance(\trueBar, x)\bigr)}{ \radiusEncBall} \le 2\cdot\frac{\max_{x\in \bb}\bigl(\distance(\trueBar, x)\bigr)}{\max_{x, y\in \barSet}\bigl(\distance(x, y)\bigr)} \le 2\cdot\frac{\sqrt{d}}{1} = 2\sqrt{d}.\qedhere$$
\end{proof}

\begin{restatable}{lemma}{syncBoxAlgResilience}\label{lem:syncBoxAlgResilience}
    \Cref{alg:box approach} can tolerate up to $t$ \byz where $t<n/3$.
\end{restatable}

\begin{proof}
    In \Cref{lem:bb-and-tb-intersect}, we discuss that $\bbi\cap\tbi\neq \emptyset$. This proof in fact requires an upper bound of $t<n/2$ on the number of \byzadj nodes. The required upper bound of $t<n/3$ on the number of \byzadj nodes comes from the implementation of reliable broadcast as a subroutine. This property is implicitly used to ensure that the set $\barSet$ is bounded in size by $\binom{n}{n-t}$ and cannot be extended to an arbitrarily large size by \byzadj parties. 
\end{proof}

This lemma concludes the proof of \Cref{thm:syncBoxAlg} and the analysis of the centroid approximation for the synchronous case.  

Observe that \Cref{thm:syncBoxAlg} directly implies \Cref{cor:synch-RB-TM} for the $RB\!-\!TM$ algorithm from~\cite{NEURIPS2021_d2cd33e9}.

\begin{corollary}\label{cor:synch-RB-TM}
    The $RB\!-\!TM$ algorithm~\cite{NEURIPS2021_d2cd33e9} with a resilience of $t<n/3$ achieves a $2\sqrt{d}$-approximation of $\trueBar$ in the synchronous setting.
\end{corollary}
\begin{proof}
    Observe that the trimmed mean computed in \cite{NEURIPS2021_d2cd33e9} corresponds to the vector $V_k$ constructed in the proof of \Cref{lem:bb-and-tb-intersect}. Thus, the trimmed mean lies in the intersection of $\tbi$ and $\bbi$ for every node $i$.
\end{proof}

\section{Asynchronous algorithms for approximate computation of the centroid}\label{sec:asynch_approximations}

In the asynchronous case, the nodes cannot rely on the fact that they receive vectors from all correct nodes. 
Since, $\byz$ could decide to not send any vectors, the nodes can only be allowed to wait for $n-t$ vectors (number of correct nodes that all send their vectors). 
However, if $\byz$ do send vectors, then these could arrive before $t$ true vectors, that would then be ignore by the receiving node since the latter only waits for $n-t$ messages before running the agreement algorithm. 
Therefore, only $n-2t$ of the vectors can be guaranteed to come from the correct nodes. Due to this restriction, also the definitions of the $\safe$ as well as the local trusted boxes need to be adjusted.

The asynchronous multidimensional approximate agreement algorithm from~\cite{mendes2015multidimensional} is also based on the $\safe$ approach. In the asynchronous case, the $\safe$ that is computed is slightly different. The smaller convex hulls of $n-2t$ nodes are considered for the intersection. This decreases the number of \byz that are tolerated to $t<n/(d+2)$. The decrease is necessary for the $\safe$ to be non-empty. The asynchronous $\safe$ is defined as follows:

\begin{definition}[Asynchronous $\safe$]\label{def:asynch_safe_area}
In the asynchronous setting, the $\safe$ of a set of vectors $\{v_i, i\in \left[n\right]\}$ that can contain up to $t$ \byzadj vectors is defined as 
    $$\safe = \bigcap_{\substack{I\subseteq \left[n\right]\\ |I|=n-2t}}\convexHull\bigl(\{v_i, i\in I\}\bigr).$$
\end{definition}

Also, the trusted box that the nodes can compute locally is reduced in size. In the following, we redefine the local trusted box for asynchronous communication:

\begin{definition}[Asynchronous local trusted box]\label{def:asynch_trusted_box}
For all received vectors from nodes $i \in [n-t]$, denote
$v_i[k]$ the $k^{th}$ coordinate of the respective vector. Assume that all $v_i[k]$ are sorted and relabeled 
as in \Cref{obs: reordering}.
Then, in the asynchronous case, the locally computed trusted box $\widetilde{\tbi}$ is defined as $$\widetilde{\tbi}[k] = \bigl[v_{t+1}[k], v_{n - 2t}[k]\bigr], \ \forall k\in [d].$$ 
\end{definition}

\subsection{Approximation through the asynchronous \texorpdfstring{$\boldsymbol{\safe}$}{safe} approach}

As before, we start by discussing the lower bound on the approximation of the asynchronous $\safe$ approach. Compared to the synchronous case, the \byzadj adversary is able to additionally hide up to $t$ true vectors from the correct nodes.  

\begin{theorem}
    The approximation ratio of the asynchronous multidimensional approximate agreement via the $\safe$ approach, where the $\safe$ is a single vector, is lower bounded by $2(d+1)$.
\end{theorem} 
\begin{proof}
    In this proof, we will extend the example in the proof of \Cref{thm:sync_safe_area}. Since the \byzadj adversary can hide up to $t$ true vectors, we can assume that the $t$ hidden nodes lie far away from the locally computed $\safe$. If they lie too far, they can increase the radius of the smallest enclosing ball and thus improve the approximation ratio. Therefore, we assume that the $t$ hidden true vectors lie in $(x,\ldots,0)$. 

    We can replace $d$ by $d+1$ in the calculations from \Cref{thm:sync_safe_area} and arrive at a lower bound of $2(d+1)$ for the approximation ratio of the asynchronous $\safe$ approach.
\end{proof}

\subsection{A \texorpdfstring{$\mathbf{2}$}{two}-approximation of the centroid via the \texorpdfstring{$\boldsymbol{\safe}$}{safe} approach}

Similar to \Cref{sec:synch_one_approximation}, we can achieve a good approximation of $\trueBar$ while satisfying the weak validity condition in the asynchronous case. 

In the asynchronous case, every node can only wait for $n-t$ other vectors. If \byzadj vectors are among the received messages, a correct node will not be able to compute an optimal approximation of $\trueBar$ locally. This is because the average of the $n-t$ received vectors does not need to be $\trueBar$. Instead, each node can only compute one possible centroid approximation from $\barSet$. \Cref{alg:asynch-two-approx} presents this idea in pseudocode. The next theorem shows that the computation of one centroid is sufficient to achieve a $2$-approximation of $\trueBar$.

\begin{algorithm}[tbh]
\caption{Asynchronous $2$-approximation of the centroid with resilience $t<n/(d+2)$}
\label{alg:asynch-two-approx}
\begin{algorithmic}[1]
    \State Each node $i$ in $[n]$ with input vector $v_i$ executes the following code:
    \Indent
        \State Broadcast own vector $v_i$ reliably to all nodes
        \Upon{Reliably receiving $n-t$ vectors $v_j$}
            \State $c_i \leftarrow$ centroid of the received $v_j$
            \State Run $\ExistingSafeArea$ with input vector $c_i$
        \EndUpon
    \EndIndent
\end{algorithmic}
\end{algorithm}

\begin{theorem}\label{thm:asynch_two_approx}
    \Cref{alg:asynch-two-approx} achieves a $2$-approximation of $\trueBar$ in an asynchronous setting while satisfying weak validity. After $O\left(\log_{1/(1-\gamma)}\left(\frac{1}{\varepsilon}\cdot\radiusEncBall\cdot\sqrt{d}\right)\right)$ asynchronous rounds, where $1/(1-\gamma) = n\binom{n}{n-f}\Big/\left(n\binom{n}{n-f}-1\right)$, the vectors of the correct nodes converge. The resilience of the algorithm is thereby $t<n/(d+2)$.    
\end{theorem}
 \begin{proof}
     When computing the centroid of $n-t$ vectors, each node $i$ computes one element of $\barSet$. This ensures that the algorithm agrees inside $\encBall$.
     The maximum distance between any vector in $\convexHull(\barSet)$ and $\trueBar$ is the diameter of $\encBall$, which is twice the radius $\radiusEncBall$. Therefore, \Cref{alg:asynch-two-approx} outputs a $2$-approximation of $\trueBar$. Since the only difference of this algorithm from \Cref{alg:safe area approach} is the preprocessing step, the remaining properties of the algorithm follow from the proof of \Cref{thm:synch-one-approx}.

     Observe that the algorithm trivially satisfies weak validity: if all noes are correct and all input vectors of these nodes are identical, each of the $n-t$ received vectors will be the same. Therefore, all nodes will pick the same input vector to run $\ExistingSafeArea$.
 \end{proof}

\begin{lemma}\label{lem:lowerBound-async-weakVal}
    In the asynchronous setting, the approximation ratio of the true centroid that can be achieved by any algorithm satisfying the weak validity property is at least 2.
\end{lemma}
\begin{proof}
    Consider the same setting as in \Cref{lem:lowerBound-sync-strongVal} and consider an algorithm that satisfies the weak validity condition. Once again, its convergence vector has to be $0$. Indeed, in the asynchronous setting, the algorithm only waits for $n-t$ vectors before deciding. If the $t$ input vectors $e_1$ are not received, the algorithm receives only the $n-t$ vectors $0$ and all those vectors could come from correct nodes (\byz are potentially undetectable). Hence, by definition of the weak validity condition, the algorithm has to converge to $0$. 

    But similarly to the setting of \Cref{lem:lowerBound-sync-strongVal}, the \byz can all have input vectors at $0$, making the true centroid $\left(\frac{t}{n-t}, 0, \dots, 0\right)$. 
    The \encBall being the same in synchronous and asynchronous settings, we can again conclude that $\radiusEncBall = \frac{1}{2}\cdot \frac{t}{n-t}$ and hence that the approximation ratio in this specific case is $2$.
\end{proof}

\subsection{A constant approximation via the \texorpdfstring{$\boldsymbol{\mda}$}{mda} approach}
In this section, we consider the asynchronous version of \Cref{alg:mda approach}. The only difference in the asynchronous version of the algorithm is that the $\mda$ is computed with only $n-2t$ nodes, see \Cref{alg:asynch mda approach} for a pseudocode.

\begin{algorithm}[tbh]
\caption{Asynchronous constant approximation of the centroid with $t<n/7$}
\label{alg:asynch mda approach}
\begin{algorithmic}[1]
\State Let $C$ be some large constant. Each node $i$ in $[n]$ with input vector $v_i$ executes the following code:
    \Indent
        \State Set local round $r \leftarrow 0$
        \While{$r<C\cdot\log_2 \bigl(\frac{1}{\varepsilon}\cdot D\bigr)$}
            \State $r=r+1$
            \State Broadcast $v_i$ reliably to all nodes
            \Upon{Reliably receiving $n-t$ vectors $v_j$}
                    \State Set $M_i = \{\text{all received vectors}\ v_j\}$
                    \State Compute $MDA(M_i, n-2t)$ 
                    \State Set new vector $v_i$ to be the centroid of $MDA(M_i, n-2t)$
            \EndUpon
        \EndWhile
    \EndIndent
\end{algorithmic}
\end{algorithm}

The analysis of $\mda$ is more involved in the asynchronous setting, since the vectors received after the first round are not inside $\encBall$ anymore. To show the approximation ratio, we first show that the smallest enclosing ball inside of which all correct input vectors of round $2$ will lie, is only a constant factor larger than $\radiusEncBall$. Then, we show that this larger ball and $\encBall$ intersect. Using these results, we provide a constant upper bound on the approximation ratio of \Cref{alg:asynch mda approach}.

In this section, we denote $S_1$ the set of centroids of $n-t$ vectors, i.e., $S_1= \barSet$, and $S_2$ the set of centroids of $n-2t$ vectors, i.e., $S_2=\bigl\{ \frac{1}{n-2t}\sum_{i\in I}v_i\ \big|\ \forall I\subseteq [n], |I|=n-2t, v_i\in V \bigr\}$. Further, we say that $V$ is the set of the input vectors of all nodes, including the \byzadj ones, at the beginning of a round. 

\begin{lemma}\label{lem:MDAboundBall}
    $\ballOf{S_2} \le  14/5\ballOf{S_1} $, where $ \ballOf{S_1}=\encBall$. 
\end{lemma}

\begin{proof}
    The maximum distance between two vectors of $S_1$ is: 
    \begin{align*}
        \frac{1}{n-t}\max_{u_i, v_j\in V}\left\Vert\sum_{i=1}^{n-t} u_i-\sum_{j=1}^{n-t} v_j\right\Vert_2
        = \frac{1}{n-t}\max_{\substack{u_i, v_j\in V\\ M,N\subseteq [n-t]\\ |M|=|N|=2t}}\left\Vert\sum_{i\in M, j\in N}(u_i - v_j)\right\Vert_2.
    \end{align*}
    Note that the centroids of $n-t$ vectors are sums that differ in at most $2t$ elements ($n$ vectors in total). 
    We use $a= \max\limits_{\substack{u_i, v_j\in V\\ M,N\subseteq [n-t]\\ |M|=|N|=2t}}\left\Vert\sum_{\substack{i\in M \\ j\in N}}(u_i - v_j)\right\Vert_2$ to denote the diameter of the vectors in $S_2$. Then, the maximum distance between two elements of $S_1$ is $\textsc{maxS1} \coloneqq \frac{1}{n-t}a$. 

    Similarly, the maximum distance between two elements of $S_2$ is 
    \begin{align*}
        \frac{1}{n-2t}\max_{u_i, v_j\in V}\left\Vert\sum_{i=1}^{n-2t} u_i-\sum_{j=1}^{n-2t} v_j\right\Vert_2
        &= \frac{1}{n-2t}\max_{\substack{u_i, v_j\in V\\ M,N\subseteq [n-t]\\ |M|=|N|=2t}}\left\Vert\sum_{i\in M, j\in N}(u_i - v_j)\right\Vert_2\\
        &\leq \frac{1}{n-2t}\cdot 2a.
    \end{align*}
    Hence, the maximum distance between two elements of $S_2$ is at most 
    \begin{align}
        \frac{1}{n-2t}\cdot 2(n-t)\cdot\textsc{maxS1} &= \frac{2n - 2t}{n-2t}\cdot \textsc{maxS1} \nonumber\\
        & \leq  \frac{2n}{n-2(n/7)}\cdot \textsc{maxS1}\label{eq:MDA-balls}\\
        & = \frac{14}{5}\cdot \textsc{maxS1}. \nonumber
    \end{align}
     Observe that \Cref{eq:MDA-balls} follows due to the assumption $t<n/7$ in the asynchronous case.
\end{proof}

In the following, we will show that the smallest enclosing ball formed by the centroids of subsets of $n-2t$ vectors intersects the smallest enclosing ball of all possible centroids.
\begin{lemma}\label{lem:MDAballsIntersect}
    The smallest enclosing ball of centroids of $n-2t$ vectors intersects $\encBall$.
\end{lemma}
\begin{proof}
    We will start by showing that $\forall c\in S_1: c\in \ballOf{S_2}$. 

    We consider the subset of $n-t$ vectors $V'\coloneqq\{v_1,\ldots,v_{n-t}\}$ of $V$ as an example and show that the centroid of this subset lies in $\ballOf{S_2}$. The analysis of all other centroids in $S_1$ follows analogously. We next consider all possible centroids of subsets of $n-2t$ vectors from $V'$, denoted $W\coloneqq \bigl\{ \frac{1}{n-2t}\sum_{i\in J}v_i\ \big|\ \forall J\subseteq [n-t], |J|=n-2t, v_i\in V \bigr\}$. Observe that $W\subseteq S_2$ by definition. In the following, we show that a convex combination of all vectors in $W$ equals $c_1\coloneqq\sum\limits_{i=1}^{n-t}v_i$.
    \begin{align*}
        \frac{1}{\binom{n-t}{n-2t}}&\sum_{\substack{K\subset[n-t]\\|K|=n-2t}}\left(\frac{1}{n-2t}\sum_{k\in K}v_k\right) \\
        &= \frac{1}{\binom{n-t}{n-2t}}\cdot\frac{1}{n-2t}\cdot(n-t)\cdot\binom{n-t-1}{n-2t-1}\cdot\left(\frac{1}{n-t}\frac{1}{\binom{n-t-1}{n-2t-1}}\sum_{\substack{K\subset[n-t]\\|K|=n-2t}}\sum_{k\in K}v_k\right) \\
        &= \frac{1}{\binom{n-t}{n-2t}}\cdot\frac{n-t}{n-2t}\cdot\frac{(n-t-1)!}{(n-2t-1)!\,t!}\cdot\left(\frac{1}{n-t}\sum_{i=1}^{n-t}v_i\right) \\
        &=\frac{1}{\binom{n-t}{n-2t}}\cdot\binom{n-t}{n-2t}\cdot\left(\frac{1}{n-t}\sum_{i=1}^{n-t}v_i\right) = \frac{1}{n-t}\sum_{i=1}^{n-t}v_i = c_1.
    \end{align*}
    Note that the first term is indeed a convex combination of the vectors in $S_2$, as there are exactly $\binom{n-t}{n-2t}$ many vectors in $W$, each weighted by $1/\binom{n-t}{n-2t}$. 
    
    The above calculation can be calculated for all $c\in S_1$ by defining $W$ with respect to the vectors that form centroid $c$. Since $c\in \ballOf{S_2}$, the two balls, $\ballOf{S_1}$ and $\ballOf{S_2}$, must have a nonempty intersection.
\end{proof}

With the above lemmas, we can now establish a constant upper bound on the approximation ratio of \Cref{alg:asynch mda approach}. 

\begin{theorem}\label{thm:asynch-mda}
    \Cref{alg:asynch mda approach} converges in $O\Bigl(\log_2 \bigl(\frac{1}{\varepsilon}\cdot D\bigr)\Bigr)$ asynchronous rounds and achieves a $10.4$-approximation of the centroid.
\end{theorem}
\begin{proof}
    Observe that the input of each correct node at the beginning of the second round is inside $\ballOf{S_2}$. Thus, the pairwise distance of any two correct vectors is bounded by $14/5\cdot\radiusEncBall$ (\Cref{lem:MDAboundBall}). In addition, due to \Cref{lem:MDAballsIntersect}, the vectors are at a distance of at most $38/5\cdot\radiusEncBall$ away from $\trueBar$.

    We can now use a similar analysis to the proof of \Cref{lem:MDAsynch}, setting $t<n/7$ in the asynchronous case~\cite{NEURIPS2021_d2cd33e9}: At the beginning of the second round, the \byzadj nodes can choose their vectors at a distance of at most $14/5\cdot\radiusEncBall$ from $\ballOf{S_2}$. The new inputs for round $3$ can be at a distance of at most 
    $$\frac{t}{n-3t}\cdot2\cdot\frac{14}{5}\radiusEncBall < \frac{7}{5}\radiusEncBall$$
    from $\ballOf{S_2}$.
    By the result from~\cite{NEURIPS2021_d2cd33e9}, we assume that the distance between the furthest two correct nodes shrinks by a factor of $2$ in each round. We repeat the above argument to receive a general upper bound on the approximation ratio of the asynchronous $\mda$ algorithm:
    $$\frac{38}{5}+\sum\limits_{i=0}^{\infty}\frac{2}{2^i}\cdot\frac{14}{5}\cdot\frac{t}{n-2t} < \frac{38}{5}+\frac{7}{5}\sum\limits_{i=0}^{\infty}\frac{1}{2^i} = 10.4$$
\end{proof}

\begin{observation}\label{obs:asynch_mda_and_safe}
    There exists a $4.8$-approximation of the centroid computed using a combination of $\mda$ and the $\safe$ approach that satisfies strong validity and $t<\min(n/7, n/(d+1))$.
\end{observation}

This observation comes from the fact that after the first round, all correct nodes hold a $2+14/5=4.8$-approximation of the centroid. By converging inside the convex hull of correct vectors using the $\safe$ approach, the approximation ratio stays $4.8$.

\begin{lemma}\label{lem:lowerBound-async-strongVal}
    In the asynchronous setting, the approximation ratio of the true centroid that can be achieved by any algorithm satisfying the strong validity property is at least 4.
\end{lemma}
\begin{proof}
    Consider the setting where $n-2t$ input vectors are $0$ and $2t$ input vectors have $0$ in all coordinates but $1$ in the first one, i.e.\ $e_1 = \left(1, 0, \dots, 0\right)$. 
    Consider an algorithm that satisfies the strong validity condition. Then, its convergence vector has to be $0$. Indeed, the algorithm only waits for $n-t$ vectors, and hence could only receive $n-2t$ vectors $0$ and $t$ vectors $e_1$. If the $t$ vectors $e_1$ were from \byz, then the algorithm would have to converge to $0$ by definition of strong validity. Since \byz are undetectable, the convergence vector is $0$. 
    
    However, the \byz can all have input vectors at $0$, in which case the true centroid is $\left(\frac{2t}{n-t}, 0, \dots, 0\right)$, which is at distance $\frac{2t}{n-t}$ of $0$. 
    Since $\left(\frac{2t}{n-t}, 0, \dots, 0\right)$ and $\left(\frac{t}{n-t}, 0, \dots, 0\right)$ are the two furthest elements of $\barSet$, $\radiusEncBall = \frac{1}{2}\cdot \frac{t}{n-t}$. Hence, the approximation ratio in this specific case is $\frac{2t}{n-t}\cdot 2\cdot \frac{n-t}{t} = 4$.
\end{proof}

\subsection{A \texorpdfstring{$\mathbf{4\sqrt{d}}$}{tau2}-approximation of the centroid with box validity}

At the beginning of \Cref{sec:asynch_approximations}, we redefined the local trusted box for the asynchronous case. Note that the local trusted box in the asynchronous case is just smaller than in the synchronous case, and is still included in $\tb$. Observe that this is not true for the centroid box. In the asynchronous case, each correct node can only compute one vector from $\barSet$, and therefore $\bbi$ would be just containing one point. In order to derive similar convergence results to the synchronous case, we need to relax the definition of $\bb$. This relaxation will lead to a worse approximation ratio in the asynchronous case.

\begin{definition}[Relaxation of $\bb$]\label{def:relaxation_CB}
    \pbb is the smallest box containing all possible centroids of $n-2t$ vectors: 
    \begin{align*}
       \pbb =  \convexHull\left(\left\{ \frac{1}{n-2t}\sum_{i\in I}v_i\ \bigg|\ \forall I\subseteq [n], |I|=n-2t \right\}\right), 
    \end{align*}
    and \pbbi is the smallest local box containing all possible centroids of $n-2t$ vectors, formed from the vectors received by node $i$:
    \begin{align*}
        \pbbi = \convexHull\left(\left\{ \frac{1}{n-2t}\sum_{v\in V}v \ \bigg|\ \forall V\subseteq M_i, |V|=n-2t \right\}\right),
    \end{align*}
    where $M_i$ is the set of vectors received by node $i$.
\end{definition}

With these definitions, we can now present an algorithm solving asynchronous approximate multidimensional agreement with box validity. The algorithm is similar to \Cref{alg:box approach}, with the difference that the nodes choose their new vector inside $\widetilde{\tbi}\cap \pbbi$.  We will assume that each node receives exactly $n-t$ vectors in each round. Though a node may receive more vectors in a round, the additional vectors can be ignored by the algorithm. \Cref{alg:box-approach-async} presents this strategy in pseudocode. 

\begin{algorithm}[tbh]
\caption{Asynchronous approximate agreement with box validity and resilience $t<n/3$}
\label{alg:box-approach-async}
\begin{algorithmic}[1]
    \State Let $C$ be some large constant. Each node $i$ in $[n]$ with input vector $v_i$ executes the following code:
    \State Set local round $r \leftarrow 0$
    \While{$r<C\cdot\log_2\left( \frac{1}{\varepsilon}\cdot\longestEdge(\tbi) \right)$}
        \State $r \leftarrow r + 1$
        \State Broadcast $(v_i, r)$ reliably to all nodes
        \Upon{Reliably receiving $n-t$ vectors $v_j$ in round $r$}
            \State Set $M_i \leftarrow$ received $(v_j, r)$ 
            \State Compute $\widetilde{\tbi}$ and $\pbbi$ from $M_i$ by excluding $t$ smallest and $t$ largest values in each coordinate
            \State Set $v_i$ to $\midpointFunction(\widetilde{\tbi}\cap \pbbi)$
        \EndUpon
    \EndWhile
\end{algorithmic}
\end{algorithm}

Before diving into the analysis of the algorithm, we make the following observations:

\begin{observation}\label{obs:relation_bb_pbb}
    The following properties hold:
    \begin{itemize}
        \item $\bb\subseteq \pbb$ (by \Cref{def:relaxation_CB}),
         \item The intersection $\widetilde{\tbi}\cap \bbi$ can be empty
        \item $\pbbi\subseteq \pbb$.
    \end{itemize}
\end{observation}

The following theorem describes the properties of \Cref{alg:box-approach-async}:

\begin{theorem}\label{thm:asynch_BoxAlgo}
    \Cref{alg:box-approach-async} achieves a $4\sqrt{d}$-approximation of $\trueBar$ in an asynchronous setting. After $O\Bigl( \log_2\left( \frac{1}{\varepsilon}\cdot\longestEdge(\tb) \right) \Bigr)$ asynchronous rounds, the vectors of the correct nodes converge. The resilience of the algorithm is thereby $t<n/3$.
\end{theorem}

We will show this theorem in several steps. Observe first that the number of rounds of the algorithm can be derived analogously to the proof of \Cref{lem:syncBoxAlgRounds}. We start by showing that the algorithm is correct and converges. From this, one can derive that the solution satisfies box validity. Therefore, \Cref{alg:box-approach-async} solves multidimensional approximate agreement. The required resilience for the analysis is $t<n/3$. We will skip the proofs of some of these statements, since they are equivalent to the proofs in \Cref{sec:synch_sqrt_d_approximation}. In the second part, we show that the approximation ratio of the algorithm is upper bounded by $4\sqrt{d}$.

\begin{lemma}\label{lem:asyncBocAlgCorrect}
    \Cref{alg:box-approach-async} is well-defined and converges.
\end{lemma}
\begin{proof}
    
We start by showing that \Cref{alg:box-approach-async} is well-defined. We therefore need to prove that $\widetilde{\tbi}\cap \pbbi \neq \emptyset$. 

Consider $\widetilde{\tbi}[k]$ and $\pbbi[k]$ to be the projections of $\widetilde{\tbi}$ and $\pbbi$ on the $k^{th}$ coordinate respectively. 
Note that $\widetilde{\tbi}$ is always computed from at least $n-t$ received messages. In order to compute $\widetilde{\tbi}$, a node would remove $t$ values ``on each side'' for each coordinate $k\in[d]$. 
For each coordinate  $k\in [d]$, we denote by $v_j[k]$ the $k^{th}$ coordinate of the vector that could be received from node $j$. 
As in previous proofs, 
we reorder the values $v_j[k]$ and reassign the indices according to \Cref{obs: reordering}. 
Now we can rewrite $\widetilde{\tbi}[k] = \bigl[v_{t+1}[k], v_{n-2t}[k]\bigr]$. 

Consider the average vector $c_\ell = \frac{1}{n-2t}\sum_{j=1}^{n-2t}v_j[k]$. 
Since it corresponds to a sum of the $k^{th}$ coordinates of a set of $n-2t$ vectors, there exists a centroid in $\barSet$ s.t.\ the $k^{th}$ coordinate is $c_\ell = \frac{1}{n-2t}\sum_{j=1}^{n-2t}v_j[k]$.
Similarly, $\exists\ c_u = \frac{1}{n-2t}\sum_{j=t+1}^{n-t}v_j[k]\in \barSet[k]$. 

Since the values $v_j[k]$ are in increasing order, $\frac{1}{n-3t}\sum_{j=t+1}^{n-2t}v_j[k] \geq c_\ell$ and $\frac{1}{n-3t}\sum_{j=t+1}^{n-2t}v_j[k] \leq c_u$. 
We define $V_k = \frac{1}{n-3t}\sum_{j=t+1}^{n-2t}v_j[k]$. Note that $V_k\in [c_\ell, c_u]\subseteq \pbbi[k]$, since $\pbbi$ is a convex polytope. 
Also, $V_k \in \bigl[v_{t+1}[k], v_{n-2t}[k]\bigr] = \widetilde{\tbi}[k]$, since it is the mean of values $v_{t+1}[k]$ to $v_{n-2t}[k]$.

Now, consider the vector $V$ s.t.\ $\forall k\in [d]$, $V[k] = V_k$. 
By definition of $\widetilde{\tbi}$ and $\pbbi$, since $V[k] \in \widetilde{\tbi}[k]\cap \pbbi[k], \forall k\in[d]$, $V\in\widetilde{\tbi}\cap \pbbi$. This concludes the proof of the algorithm being well-defined.

The proof of convergence is the same as in the proof of \Cref{lem: boxAlgCorrect} since only $\tb$ was used to prove convergence and not $\bb$. 
\end{proof}

To prove an upper bound on the approximation ratio of \Cref{alg:box-approach-async}, we need the following intermediate lemma, which allows us to use \Cref{lem:syncBoxAlgApprox} for the upper bound proof. 

\begin{lemma}\label{lemma: longest edges ratio is two}
    The diagonal of \pbb is at most twice the diagonal of $\bb$.
\end{lemma}

\begin{proof}
Note that the boxes \bb and \pbb cannot be computed locally, but the locally computed boxes \bbi and \pbbi are included in \bb and \pbb respectively. 
We consider coordinate $k$, and assume that the values $v_i[k]$ are indexed so that they are in increasing order
(see \Cref{obs: reordering}).
Denote $\mathrm{Edge}\bigl(\bb[k]\bigr)$ the length of the interval $\bb[k]$. 
Then, the length of $\bb[k]$ is, 
    \begin{align*}
        \mathrm{Edge}\bigl(\bb[k]\bigr)&= \frac{1}{n-t}\sum_{i=t+1}^{n}v_{i}[k] - \frac{1}{n-t}\sum_{i=1}^{n-t}v_i[k]
        = \frac{1}{n-t}\sum_{i=1}^{t}v_i[k] + \frac{1}{n-t}\sum_{i=n-t+1}^{n}v_{i}[k]\\
        &\geq \frac{1}{n-t}\left[\sum_{i=1}^{2t}v_i[k] + \sum_{i=n-2t+1}^{n}v_{i}[k]\right].
    \end{align*}
And the length of $\pbb[k]$ is, 
\begin{align*}
        \mathrm{Edge}\bigl(\pbb[k]\bigr)&= \frac{1}{n-2t}\sum_{i=2t+1}^{n}v_{i}[k] - \frac{1}{n-2t}\sum_{i=1}^{n-2t}v_i[k]\\
    &= \frac{1}{n-2t}\left[\sum_{i=1}^{2t}v_i[k] + \sum_{i=n-2t+1}^{n}v_{i}[k]\right].
    \end{align*}

Denote  $a = \sum_{i=1}^{2t}v_i[k]$ and $b =\sum_{i=n-2t+1}^{n}v_{i}[k]$. 
Then, 
\begin{align*}
    \frac{\mathrm{Edge}\bigl(\pbb[k]\bigr)}{\mathrm{Edge}\bigl(\bb[k]\bigr)} 
    \leq \frac{\frac{1}{n-2t}\cdot a + \frac{1}{n-2t}\cdot b}{\frac{1}{n-t}\cdot a + \frac{1}{n-t}\cdot b}
    \leq \frac{\frac{2}{n-t}\cdot a + \frac{2}{n-t}\cdot b}{\frac{1}{n-t}\cdot a + \frac{1}{n-t}\cdot b}
     = 2,
\end{align*}
since $n\geq 3t \Leftrightarrow \frac{1}{n-2t}\leq \frac{2}{n-t}$.

Observe that, since this inequality is true for every coordinate $k$, it also holds for the diagonals of the boxes $\bb$ and $\pbb$. 

\end{proof}

\begin{lemma}\label{lem:asyncBoxAlgApprox}
    The approximation ratio of \Cref{alg:box-approach-async} is upper bounded by $4\sqrt{d}$. 
\end{lemma}

\begin{proof}
We can lower bound the radius of the enclosing ball by $\max_{x, y\in \barSet}\bigl(\distance(x, y)\bigr)/2$, as in the proof of \Cref{lem:syncBoxAlgApprox}. 

Observe that \Cref{alg:box-approach-async} does not converge inside $\bb$, but inside $\pbb$. 
The furthest \Cref{alg:box-approach-async} can get from $\trueBar$ is then $\max_{x\in \pbb}(\distance\bigl(\trueBar, x)\bigr)$.
Assume w.l.o.g. that the longest edge of $\bb$ is $1$.
Since $\bb\subseteq\pbb$ (\Cref{obs:relation_bb_pbb}) and the diagonal of $\pbb$ is at most twice the diagonal of $\bb$ (\Cref{lem: boxAlgCorrect}), the furthest distance is upper bounded by $$\max_{x\in \pbb}\bigl(\distance(\trueBar, x)\bigr) \le 2\max_{x\in \bb}\bigl(\distance(\trueBar, x)\bigr).$$

Therefore, the approximation ratio is upper bounded by
$$\frac{\max_{x\in \pbb}\bigl(\distance(\trueBar, x)\bigr)}{ \radiusEncBall} \le 2\cdot\frac{2\max_{x\in \bb}\bigl(\distance(\trueBar, x)\bigr)}{\max_{x, y\in \barSet}\bigl(\distance(x, y)\bigr)} \le 4\cdot\frac{\sqrt{d}}{1} = 4\sqrt{d}.$$
 
\end{proof}

\begin{corollary}\label{cor:asynch-RB-TM}
    The $RB\!-\!TM$ algorithm~\cite{NEURIPS2021_d2cd33e9} with a resilience of $t<n/3$ achieves a $2\sqrt{d}$-approximation of $\trueBar$ in the synchronous setting.
\end{corollary}
\begin{proof}
    Observe that the trimmed mean of $n-2t$ values computed in~\cite{NEURIPS2021_d2cd33e9} corresponds to the vector $V_k$ constructed in the proof of \Cref{lem:asyncBocAlgCorrect}. Thus, the trimmed mean lies in the intersection of $\widetilde{\tbi}$ and $\widetilde{\bbi}$ for every node $i$.
\end{proof}

\section{Open Questions}\label{sec:open-questions}

Observe that the presented results for box validity in this paper are not tight. The upper and lower bounds on the approximation ratio for the multidimensional agreement algorithm with box validity differ by a factor of $\sqrt{d}$. This leaves the following open questions for future work:

\begin{openquestion}
    Does there exist a multidimensional approximate agreement algorithm that satisfies box validity and that provides a $o(\sqrt{d})$-approximation of the centroid?
\end{openquestion}  

\begin{openquestion} 
    Alternatively, is there a lower bound of $\Omega(\sqrt{d})$ on the approximation ratio of any deterministic algorithm that satisfies box validity?
\end{openquestion}

Note that we use a constructive proof in this paper in order to show the $\Omega(\sqrt{d})$ upper bound on the approximation ratio. This construction uses strict assumptions on the locations of the centroids. In $2$ and $3$ dimensions, it is possible to construct inputs that result in the desired centroids. However, $2$ and $\sqrt{d}$ are only a small constant factor away from each other in these examples, and the investigated constructions do not generalize well to $d$ dimensions.

In addition to focusing on the particular validity condition, other validity conditions may be of interest:

\begin{openquestion} 
    Consider a validity condition that is stronger than the strong validity condition for one-dimensional multi-valued agreement and weaker than the box validity condition. Is there a validity condition that allows us to compute an $O(1)$-approximation of the centroid? 
\end{openquestion} 
\begin{openquestion} 
    Consider relaxations of convex validity that are stronger than box validity. Are there relaxations that would provide an $o(d)$-approximation of the centroid while keeping $t<n/3$ and the polynomial local computation time?
\end{openquestion}

The last question has been answered negatively by Xiang and Vaidya~\cite{xiang_et_al:LIPIcs:2017:7095} for relaxations of the convex validity conditions that are based on projections onto lower dimensions. The authors showed that this relaxation has the same resilience as the standard multidimensional approximate agreement protocol.

\section{Acknowledgments}
We would like to thank Manish Kumar, Stefan Schmid, Jukka Suomela, and Jara Uitto for
useful discussions. We would also like to thank the anonymous reviewers for the very helpful
feedback they have provided for previous versions of this work. This research was supported
by the Academy of Finland, Grant 334238, and by the Austrian Science Fund (FWF), Grant I
4800-N (ADVISE).

\bibliography{references.bib}

\end{document}